\documentclass[reprint, amsmath,amssymb,showpacs,floatfix,longbibliography, aps, twocolumn, superscriptaddress,nofootinbib]{revtex4-1}
\usepackage{CJKutf8}

\usepackage{amsthm}
\newtheorem{theorem}{Theorem}
\newtheorem{corollary}{Corollary}[theorem]

\newtheorem{example}{Example}

\usepackage{multirow}
\usepackage[utf8]{inputenc}
\usepackage{amsmath,amssymb,amsfonts,stmaryrd}
\usepackage{physics}
\usepackage{dsfont}
\usepackage{graphicx}
\usepackage{xcolor}
\usepackage{colortbl}
\usepackage{graphicx}
\usepackage{cancel}
\usepackage{natbib}
\usepackage{color}
\usepackage{tikz}
\usepackage{mathrsfs}
\usepackage{relsize}

\usepackage[all]{xy}
\usepackage{yhmath}
\usepackage{blkarray}
\usepackage{tabularx}
\usepackage{array}
\usepackage{makecell}

\usepackage{diagbox}
\usepackage{algorithm}
\usepackage{algorithmic}
\usepackage{lipsum}
\newcommand*\colvec[3][]{
    \begin{pmatrix}\ifx\relax#1\relax\else#1\\\fi#2\\#3\end{pmatrix}
}

\usepackage{amsmath}
\usepackage{bm} %allows for bold math type
\usepackage{subfigure} %allows for 2x2 figures
\usepackage{environ}
\NewEnviron{eqs}{%
\begin{equation}\begin{split}
    \BODY
\end{split}\end{equation}}
\usepackage{url}
\usepackage[margin=0.75in]{geometry}
\usepackage{scalerel}
\usepackage{stackengine,wasysym}
\newcommand\wwide[1]{\ThisStyle{%
  \setbox0=\hbox{$\SavedStyle#1$}%
  \stackengine{-.1\LMpt}{$\SavedStyle#1$}{%
    \stretchto{\scaleto{\SavedStyle\mkern.2mu\AC}{.5150\wd0}}{.6\ht0}%
  }{O}{c}{F}{T}{S}%
}}

\usepackage{hyperref}
\hypersetup{colorlinks=true,citecolor=blue,linkcolor=blue, urlcolor=blue, breaklinks=true}
\hypersetup{linktocpage}

\allowdisplaybreaks

\usepackage{mathtools}

\newcommand{\ZZ}{{\mathbb Z}}

\newcommand{\mX}{{\mathcal{X}}}

\graphicspath{{./Figures/}}

\newcolumntype{C}{>{\centering\arraybackslash}p{1em}} 

\newcommand{\change}{\color{black}}

\begin{document}

\author{Zijian Liang}
\affiliation{International Center for Quantum Materials, School of Physics, Peking University, Beijing 100871, China}

\author{Ke Liu}
\affiliation{Hefei National Research Center for Physical Sciences at the Microscale and School of Physical Sciences, University of Science and Technology of China, Hefei 230026, China}
\affiliation{Shanghai Research Center for Quantum Science and CAS Center for Excellence in Quantum Information and Quantum Physics, University of Science and Technology of China, Shanghai 201315, China}

\author{Hao Song}
\affiliation{Institute of Theoretical Physics, Chinese Academy of Sciences, Beijing 100190, China}

\author{Yu-An Chen}
\email[E-mail: ]{yuanchen@pku.edu.cn}
\affiliation{International Center for Quantum Materials, School of Physics, Peking University, Beijing 100871, China}

\date{\today}
\title{Generalized toric codes on twisted tori for quantum error correction}

\begin{abstract}
The Kitaev toric code is widely considered one of the leading candidates for error correction in fault-tolerant quantum computation. However, direct methods to increase its logical dimensions, such as lattice surgery or introducing punctures, often incur prohibitive overheads. In this work, we introduce a ring-theoretic approach for efficiently analyzing topological CSS codes in two dimensions, enabling the exploration of generalized toric codes with larger logical dimensions on twisted tori.
Using Gr\"obner bases, we simplify stabilizer syndromes to efficiently identify anyon excitations and their geometric periodicities, even under twisted periodic boundary conditions. Since the properties of the codes are determined by the anyons, this approach allows us to directly compute the logical dimensions without constructing large parity-check matrices.
Our approach provides a unified method for finding new quantum error-correcting codes and exhibiting their underlying topological orders via the Laurent polynomial ring. This framework naturally applies to bivariate bicycle codes. For example, we construct optimal weight-6 generalized toric codes on twisted tori with parameters $[[ n, k, d ]]$ for $n \leq 400$, yielding novel codes such as $[[120,8,12]]$, $[[186,10,14]]$, $[[210,10,16]]$, $[[248, 10, 18]]$, $[[254, 14, 16]]$,  $[[294, 10, 20]]$, {\change $[[310, 10, \leq 22]]$}, and $[[340, 16, 18]]$.
Moreover, we present a new realization of the {\change$[[360, 12, \leq 24]]$} quantum code using the $(3,3)$-bivariate bicycle code on a twisted torus defined by the basis vectors $(0,30)$ and $(6,6)$, improving stabilizer locality relative to the previous construction.
These results highlight the power of the topological order perspective in advancing the design and theoretical understanding of quantum low-density parity-check (LDPC) codes.
\end{abstract}

\maketitle

\tableofcontents
% \bigskip

%%%%%%%%%%%%%%%%%%%%%%%%%%%%%%%%%%%%%%%%%%%%%%%%%%%%%%%%%
\section{Introduction}

Quantum error correction is essential for scalable quantum computation~\cite{Shor1995Scheme, Steane1996Error, Knill1997Theory, gottesman1997stabilizer, kitaev2003fault}. Among the various quantum error-correcting codes developed, the Kitaev toric code is one of the most favorable candidates for practical implementation due to its high threshold~\cite{bravyi1998quantum, dennis2002topological, semeghini2021probing, Verresen2021PredictionTC, breuckmann2021quantum, bluvstein2022quantum, google2023suppressing, Google2023NonAbelian, Google2024surface, iqbal2023topological, iqbal2024NonAbelian, Cong2024EnhancingTO}.
Recently, bivariate bicycle (BB) codes have been shown to yield promising quantum error-correcting codes on small tori, in some cases performing up to an order of magnitude better than the Kitaev toric code~\cite{Bravyi2024HighThreshold, wang2024coprime, Wang2024Bivariate, tiew2024low, wolanski2024ambiguity, gong2024toward, maan2024machine, cowtan2024ssip, shaw2024lowering, cross2024linear, voss2024multivariate, berthusen2025toward, eberhardt2024logical, lin2025single}.
This progress is particularly exciting, as high-distance quantum low-density parity-check (LDPC) codes can exhibit a substantial reduction in the logical error rate once the physical error rate is below the threshold. Consequently, the ratio of physical to logical qubits can be significantly reduced while maintaining comparable error suppression~\cite{Bravyi2024HighThreshold}. These characteristics make these quantum LDPC codes appealing for near-term experimental implementations. As a result, there has been growing interest in developing efficient methods to analyze and characterize these codes.

Meanwhile, any two-dimensional translation-invariant Pauli stabilizer code over $\mathbb{Z}_2$ qubits satisfying the topological order condition~\cite{bravyi2010topological, bravyi2011short}
can be transformed by a finite-depth quantum circuit into a direct sum of the Kitaev toric codes and trivial stabilizers (product states)~\cite{bombin_Stabilizer_14, haah_module_13, haah2016algebraic, haah_classification_21, Chen2023Equivalence, ruba2024homological}.
Accordingly, concepts from topological order and topological quantum field theory (TQFT) \cite{dijkgraaf1990topological, Wen1993Topologicalorder, kitaev2006anyons, bombin2006topological, Levin2006Detecting, Chen2011Complete,  levin2012braiding, chen2012symmetry, Jiang2012Identifying, Cincio2013Characterizing, gu2014effect, gu2014lattice, jian2014layer, bombin2015gauge, wang2015topological, Ye2015Vortex, yoshida2016topological, ye2016topological, Kapustin2017Higher, Chen2018Exactbosonization, Lan2018Classification, wang2018towards, cheng2018loop, Chan2018Borromean, Chen2019Free, Han2019GeneralizedWenZee, Wang2019Anomalous, Chen2019Bosonization, Lan2019Classification, Chen2020Exactbosonization, wang2020construction, Chen2021Disentangling, Barkeshli2022Classification, Johnson2022Classification, ellison2022pauli, Chen2023HigherCup, chen2023exactly, maissam2023codimension, Kobayashi2024CrossCap, Maissam2024Higher, Maissam2024Highergroup, kobayashi2024universal, hsin2024classifying}—including anyons, fusion rules, topological spins, braiding statistics, partition functions (ground state degeneracy), and Wilson lines (logical operators)—can be directly applied.
Moreover, TQFTs, with their inherent robustness to local perturbations, naturally satisfy the quantum error-correcting criteria and can be treated as error-correcting codes.
These theoretical insights can enhance our understanding of bivariate bicycle codes and provide strategies for designing novel quantum error-correcting codes.

\begin{figure*}[t]
    \centering
    \hfill
    \includegraphics[scale=0.095]{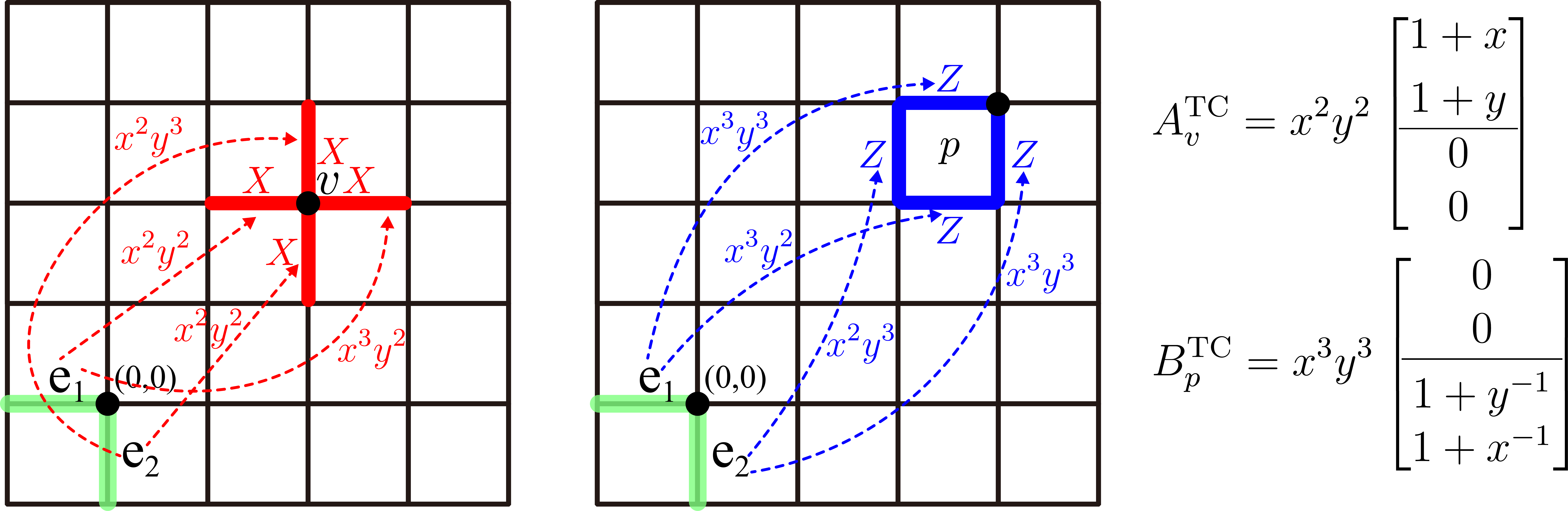}
    \hspace{1em}
    \caption{\change
    Polynomial representation of Pauli operators~\cite{haah2016algebraic}. We choose two edges $e_1$ and $e_2$ at the original as the unit cell and label their Pauli $X$ and $Z$ operators by 4-dimensional vectors, as in Eq.~\eqref{eq: X1 Z1 X2 Z2 definition}.
    The translation group $\mathbb{Z}^2$, generated by $x$ and $y$, acts by sending an operator at the origin to the site $(m,n)$ via multiplying its vector by the monomial $x^n y^m$.
    As an illustration, the $A^\mathrm{TC}_v$ and $B^\mathrm{TC}_p$ stabilizers of the Kitaev toric code are shown, with monomials such as $x^2y^2$ and $x^3y^3$ indicating their positions relative to the origin. In this way, all Pauli operators form a module over the Laurent polynomial ring $R=\mathbb{Z}_2[x,x^{-1},y,y^{-1}]$.
    }
    \label{fig: polynomial convention}
    \vspace{0.5em}
    \centering
    \includegraphics[scale=0.095]{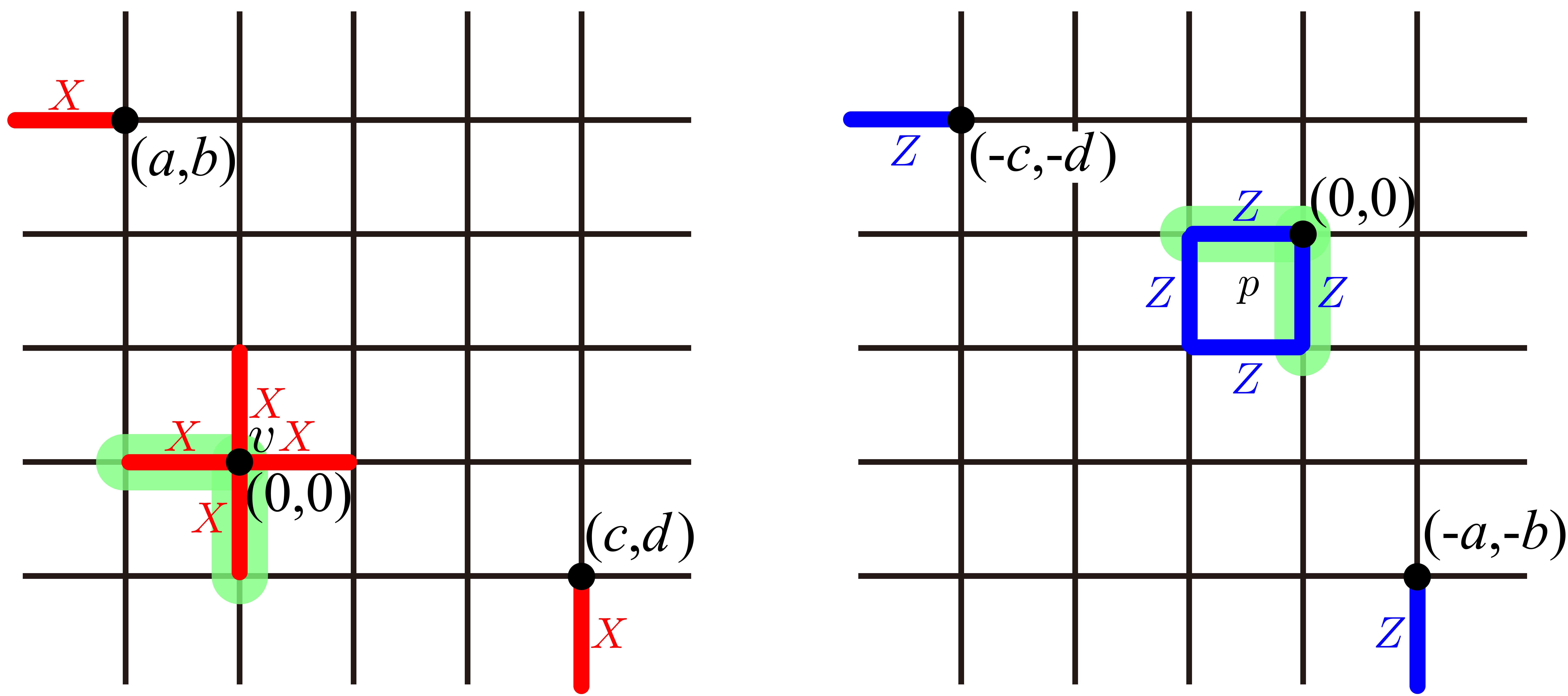}
    \caption{The $A_v$ and $B_p$ stabilizers of the generalized toric codes, parameterized by the polynomials $f(x,y) = 1 + x + x^a y^b$ and $g(x,y) = 1 + y + x^c y^d$ {\change in Eq.~\eqref{eq: stabilizer}}. The green-shaded region represents the unit cell at the origin used to generate the Pauli module over the Laurent polynomial ring~\cite{haah_module_13}. Stabilizers are specified by the integers $(a, b, c, d)$. Even when the stabilizers are identical, their implementation on different lattices yields various quantum LDPC codes.
    For instance, we later demonstrate that the $(-1,3,3,-1)$-generalized toric code (Example~\ref{example: -1 3 3 -1 generalized TC on infinite plane}), also known as the $(3,3)$-bivariate bicycle (BB) code~\cite{Bravyi2024HighThreshold, liang2024operator}, produces the $[[72, 8, 8]]$, $[[108, 8, 10]]$, $[[144, 12, 12]]$, $[[162, 8, 14]]$, $[[180, 8, 16]]$, $[[192, 8, 16]]$, $[[234, 8, 18]]$, $[[270, 8, 20]]$, {\change $[[282, 4, \leq24]]$, and $[[360, 12, \leq24]]$} quantum LDPC codes.
    }
    \label{fig: generalized_toric_code}
\end{figure*}

In this paper, using the framework of topological order, we develop a ring-theoretic approach to analyze the properties of two-dimensional topological CSS codes {\change in the Laurent polynomial formalism (Fig.~\ref{fig: polynomial convention})}. This method enables the efficient construction of new quantum LDPC codes, as summarized in Tables~\ref{tab: n_k_d 1}, \ref{tab: n_k_d 2}, \ref{tab: n_k_d 3}, and~\ref{tab: n_k_d 4}. We present the optimal $[[n, k, d]]$ with $n \leq 400$, for generalized toric codes (Fig.~\ref{fig: generalized_toric_code}) on twisted tori (Fig.~\ref{fig: twisted torus}).
For each $[[n, k, d]]$, there are typically multiple solutions for stabilizers and lattice configurations that can generate the same parameters. The polynomials and lattice vectors presented in these tables correspond to those with the most localized stabilizers, offering a more feasible construction.

We emphasize that the use of twisted tori facilitates the construction of stabilizers with more localized support compared to previous methods. For instance, the stabilizers in the $[[360, 12, \leq24]]$ code in Ref.~\cite{Bravyi2024HighThreshold} have a range of 9, as determined by the polynomial degrees. In contrast, the $[[360, 12, \leq24]]$ code presented in Table~\ref{tab: n_k_d 4} requires stabilizers with a reduced range of just 3, making it more practical for experimental implementation.

%%%%%%%%%%%%%%%%%%%%%%%%%%%%%%%%%%%%%%%%%%%%%%%%%%%%%%%%%
\section{Algebraic methods for error-correcting codes}

We adopt an algebraic approach to analyze quantum codes on lattices~\cite{haah2016algebraic}. By incorporating Laurent polynomial rings, we can extract the topological order associated with Pauli stabilizer codes~\cite{liang2023extracting}. Extending this framework, we introduce a ring-theoretic technique that simplifies computations for CSS codes. In particular, we employ Gr\"obner basis methods to systematically classify anyons in these topological orders, from which code properties naturally follow.  
The ground state degeneracy (GSD) on a torus is directly determined by the number of anyons. Specifically, the partition function of the (2+1)D TQFT satisfies~\cite{Witten1989Jones, Wen1995Topological, watanabe2023ground}:  
\begin{equation}
    Z(T^2 \times S^1) = \mathrm{GSD}_{T^2} = |\mathcal{A}|,
\label{eq: gsd and anyon number}
\end{equation}
where $\mathcal{A}$ denotes the corresponding anyon theory (unitary modular tensor category)~\cite{rowell2009classification, wang2010topological, Wang2022in, Barkeshli2022Classification, Plavnik2023Modular}. The logical operators of the code are realized as Wilson line operators (anyon string operators) wrapping around the noncontractible cycles of the torus.

We focus on the square lattice for simplicity.
Our goal is to analyze anyons in topological CSS codes, and we will demonstrate that the Gr\"obner basis technique provides an efficient way for computation.
{\change
First, we briefly review the polynomial representation of Pauli operators~\cite{haah_module_13, haah2016algebraic, liang2023extracting}, with a slightly modified convention. On the square lattice, we choose a unit cell (indicated in Fig.~\ref{fig: polynomial convention}) consisting of two edges and represent their Pauli operators as 4-dimensional vectors:
\begin{equation}
    \mathcal{X}_1 =
    \begin{bmatrix}
        1 \\
        0 \\
        \hline
        0 \\
        0
    \end{bmatrix},~
    \mathcal{Z}_1 =
    \begin{bmatrix}
        0 \\
        0 \\
        \hline
        1 \\
        0
    \end{bmatrix},~
    \mathcal{X}_2 =
    \begin{bmatrix}
        0 \\
        1 \\
        \hline
        0 \\
        0
    \end{bmatrix},~
    \mathcal{Z}_2 =
    \begin{bmatrix}
        0 \\
        0 \\
        \hline
        0 \\
        1
    \end{bmatrix}.
\label{eq: X1 Z1 X2 Z2 definition}
\end{equation}
Pauli operators on translated edges are obtained by multiplying these basis vectors by the monomial $x^n y^m$ ($n$ and $m$ could be negative), which implements a shift of $n$ steps in the $x$-direction and $m$ steps in the $y$-direction. In this language, the product of two Pauli operators corresponds simply to the sum of their four-dimensional vectors. Examples are shown in Fig.~\ref{fig: polynomial convention}. The vectors of Pauli operators form a module over the {\bf Laurent polynomial ring} $R= \mathbb{Z}_2[x, y, x^{-1}, y^{-1}]$.

}

\renewcommand{\arraystretch}{1.2}
\setlength{\tabcolsep}{0pt} % Reduce column separation
\begin{table}[t]
\centering
\definecolor{mycolor1}{RGB}{255, 200, 100}  
\definecolor{mycolor2}{RGB}{200, 100, 200}
\definecolor{mycolor3}{RGB}{100, 180, 150}
\definecolor{mycolor4}{RGB}{100, 100, 150}
\definecolor{mycolor5}{RGB}{174, 217, 69}
\definecolor{mycolor6}{RGB}{250, 50, 200}
\definecolor{mycolor7}{RGB}{50, 250, 250}
\begin{tabular}{|c|c|c|c|c|c|}
\hline
$[[n,k,d]]$ & \makecell{$~f(x,y)=~$ \\ $1+x+...$}
& \makecell{$~g(x,y)=~$ \\ $1+y+...$} & $\vec{a}_1$   &$\vec{a}_2$   
&$\frac{kd^2}{n}$ 
\\ \hline

\rowcolor{cyan!16}$[[12,4,2]]$ & $xy$&$xy$ &
$(0,3)$&$(2,1)$   
&1.33 
\\ \hline

$[[14,6,2]]$ & $y$&$x$ &
$(0,7)$&$(1,2)$   
&1.71 
\\ \hline

\rowcolor{cyan!16}$[[18,4,4]]$ & $xy$&$xy$ &
$(0,3)$&$(3,0)$   
&\pmb{3.56} 
\\ \hline

\rowcolor{cyan!16}$[[24,4,4]]$ & $xy$&$xy$ &
$(0,3)$&$(4,2)$   
&2.67 
\\ \hline

$[[28,6,4]]$ & $x^{-1}y$&$xy$ &
$(0,7)$&$(2,3)$   
&3.43 
\\ \hline

\rowcolor{magenta!15}$[[30,4,6]]$ & $\pmb{x^2}$&$\pmb{x^2}$ &
$(0,3)$&$(5,1)$   
&\pmb{4.8} 
\\ \hline

$[[36,4,6]]$ & $x^{-1}$&$y^{-1}$ &
$(0,9)$&$(2,4)$   
&4.0 
\\ \hline

\rowcolor{gray!20}$[[42,6,6]]$ & $\pmb{xy}$&$\pmb{xy^{-1}}$ &
$(0,7)$&$(3,2)$   
&\pmb{5.14} 
\\ \hline

\rowcolor{magenta!15}$[[48,4,8]]$ & $\pmb{x^2}$&$\pmb{x^2}$ &
$(0,3)$&$(8,1)$   
&\pmb{5.33} 
\\ \hline

%Added ~
$[[54,8,6]]$ & $x^{-1}$&$x^3y^2$ &
$(0,3)$&$(9,0)$   
&~5.33~ 
\\ \hline

\rowcolor{mycolor1!80}$[[56,6,8]]$ & $\pmb{y^{-2}}$&$\pmb{x^{-2}}$ &
$(0,7)$&$(4,3)$   
&\pmb{6.86}
\\ \hline

$[[60,8,6]]$ & $y^{-2}$&$x^2$ &
$(0,10)$&$(3,3)$   
&4.8 
\\ \hline

$[[\pmb{62,10,6}]]$ &$x^{-1}y$&$x^{-1}y^{-1}$&
$(0,31)$&~$(1,13)$~  
&5.81 
\\ \hline

\rowcolor{teal!40}~$[[66,4,10]]$~ & $x^{-2}y^{-1}$&$x^2y$ &
$(0,3)$&$(11,2)$   
&6.06 
\\ \hline

\rowcolor{gray!20}$[[70,6,8]]$ & $xy$&$xy^{-1}$ &
$(0,7)$&$(5,1)$   
&5.49 
\\ \hline

\rowcolor{mycolor7!70}$[[72,8,8]]$ & $\pmb{x^{-1}y^3}$&$\pmb{x^3y^{-1}}$ &
$(0,12)$&$(3,3)$   
&\pmb{7.11} 
\\ \hline

\rowcolor{teal!40}$[[\pmb{78,4,10}]]$ & $x^{-2}y^{-1}$&$x^2y$ &
$(0,3)$&$(13,1)$   
&5.13 
\\ \hline

\rowcolor{yellow!30}$[[84,6,10]]$ & $\pmb{x^{-2}}$&$\pmb{x^{-2}y^2}$ &
~$(0,14)$~ &~$(3,-6)$~   
&\pmb{7.14} 
\\ \hline
% $[[84,6,10]]$ & $x^{-2}$&$x^{-2}y^2$ &
% $(0,14)$&$(3,8)$   
% &7.14 
% \\ \hline

\rowcolor{green!40}$[[90,8,10]]$ & $\pmb{x^{-1}y^{-3}}$&$\pmb{x^3y^{-1}}$ &
$(0,15)$&$(3,-6)$   
&\pmb{8.89} 
\\ \hline
% $[[90,8,10]]$ & $x^2y^3$&$x^{-3}y^2$ &
% $(0,15)$&$(3,9)$   
% &8.89 
% \\ \hline

\rowcolor{brown!40}$[[96,4,12]]$ & $x^{-2}y$&$xy^{-2}$ &
$(0,12)$&$(4,2)$   
&6 
\\ \hline

\rowcolor{red!35}$[[98,6,12]]$ & $\pmb{x^{-1}y^2}$&$\pmb{x^{-2}y^{-1}}$ &
$(0,7)$&$(7,0)$   
&8.82 
\\ \hline

\rowcolor{mycolor2!80}$[[102,4,12]]$ & $x^{-3}y$&$x^3y^2$ &
$(0,3)$&$(17,2)$   
&5.65 
\\ \hline

\rowcolor{green!40}$[[108,8,10]]$ & $\pmb{x^{-1}y^{-3}}$&$\pmb{x^3y^{-1}}$ &
$(0,9)$&$(6,0)$   
&7.41 
\\ \hline

\rowcolor{mycolor7!70} ~$[[108,8,10]]$~ & $\pmb{x^{-1}y^3}$&$\pmb{x^3y^{-1}}$ &
$(0,9)$&$(6,0)$   
&7.41 
\\ \hline

\end{tabular}
\caption{Optimal weight-6 generalized toric codes $[[n,k,d]]$ with $n \leq 110$. The stabilizer code is defined by $f(x,y) = 1 + x + x^a y^b$ and $g(x,y) = 1 + y + x^c y^d$, as depicted in Fig.~\ref{fig: generalized_toric_code}. The second and third columns correspond to the terms $x^a y^b$ and $x^c y^d$, respectively. The basis vectors $\vec{a}_1$ and $\vec{a}_2$ define the twisted torus, illustrated in Fig.~\ref{fig: twisted torus}. Rows with the same color share identical stabilizers, i.e., the same polynomials $f(x,y)$ and $g(x,y)$, but implemented on different lattices. Bold $[[\pmb{n, k, d}]]$ denote newly discovered codes in this work, while bold polynomials highlight the novel stabilizers we found, which are more localized compared to previous constructions.
{\change Code distances are computed exactly from the integer programming approach~\cite{landahl2011fault, Bravyi2024HighThreshold}.}
All results were obtained using a personal computer.  
}
\label{tab: n_k_d 1}
\end{table}

%%%%%%%%%%%%%%%%%%%%%%%%%%%%%%%%%%%%%%%%%%%%%%%%%%%%%%%%%%%
\renewcommand{\arraystretch}{1.2}
\setlength{\tabcolsep}{0pt} % Reduce column separation
\begin{table}[t]
\centering
\definecolor{mycolor1}{RGB}{255, 200, 100}  
\definecolor{mycolor2}{RGB}{200, 100, 200}
\definecolor{mycolor3}{RGB}{100, 180, 150}
\definecolor{mycolor4}{RGB}{100, 100, 150}
\definecolor{mycolor5}{RGB}{174, 217, 69}
\definecolor{mycolor6}{RGB}{250, 50, 200}
\definecolor{mycolor7}{RGB}{50, 250, 250}
\definecolor{mycolor8}{RGB}{250, 250, 50}
\begin{tabular}{|c|c|c|c|c|c|}
\hline
$[[n,k,d]]$ & \makecell{$~f(x,y)=~$ \\ $1+x+...$}
& \makecell{$~g(x,y)=~$ \\ $1+y+...$} & $\vec{a}_1$   &$\vec{a}_2$   
&$\frac{kd^2}{n}$ 
\\ \hline

\rowcolor{red!35}$[[112,6,12]]$ & $x^{-1}y^2$&$x^{-2}y^{-1}$ &
$(0,7)$&$(8,2)$   
&7.71 
\\ \hline

$[[114,4,14]]$ & $x^{-3}y$&$x^{-5}$ &
$(0,3)$&$(19,1)$   
&6.88 
\\ \hline

\rowcolor{blue!30}$[[\pmb{120,8,12}]]$ & $x^{-2}y$&$xy^2$ &
$(0,10)$&$(6,4)$   
&\pmb{9.6} 
\\ \hline

% Added~
\rowcolor{red!35}$~[[\pmb{124,10,10}]]~$ & $x^{-1}y^2$&$x^{-2}y^{-1}$ &
$~(0,31)~$&$~(2,-12)~$   
&8.06 
\\ \hline
% $[[124,10,10]]$ & $x^{-1}y^2$&$x^{-2}y^{-1}$ &
% $(0,31)$&$(2,19)$   
% &8.06 
% \\ \hline

$[[126,12,10]]$ & $\pmb{x^{-1}y^{-2}}$&$\pmb{xy^{-1}}$ &
$(0,9)$&$(7,3)$   
&9.52 
\\ \hline

\rowcolor{mycolor1!80}$[[132,4,14]]$ & $y^{-2}$&$x^{-2}$ &
$(0,33)$&$(2,-7)$   
&5.94 
\\ \hline
% $[[132,4,14]]$ & $y^{-2}$&$x^{-2}$ &
% $(0,33)$&$(2,26)$   
% &5.94 
% \\ \hline

\rowcolor{mycolor2!80}$[[\pmb{138,4,14}]]$ & $x^{-3}y$&$x^3y^2$ &
$(0,3)$&$(23,2)$   
&5.68 
\\ \hline

\rowcolor{yellow!30}$[[\pmb{140,6,14}]]$ & $x^{-2}$&$x^{-2}y^2$ &
$(0,7)$&$(10,1)$   
&8.4 
\\ \hline

\rowcolor{green!40}$[[144,12,12]]$ & $\pmb{x^{-1}y^{-3}}$&$\pmb{x^3y^{-1}}$ &
$(0,12)$&$(6,0)$   
&\pmb{12} 
\\ \hline

\rowcolor{mycolor7!70}$[[144,12,12]]$ & $x^{-1}y^3$&$x^3y^{-1}$ &
$(0,12)$&$(6,0)$   
&\pmb{12} 
\\ \hline

$[[\pmb{146,18,4}]]$ &$y^2$&$x^{-4}y$&
$(0,73)$&$(1,16)$ 
&1.97
\\ \hline

\rowcolor{blue!30}$[[\pmb{150,8,12}]]$ & $x^{-2}y$&$xy^2$ &
$(0,25)$&$(3,7)$   
&7.68 
\\ \hline

$[[{154,6,16}]]$ & $\pmb{x^{-1}y^2}$&$\pmb{y^{-4}}$ &
$(0,77)$&$(1,16)$   
&9.97
\\ \hline

\rowcolor{brown!40}$[[\pmb{156,4,16}]]$ & $x^{-2}y$&$xy^{-2}$ &
$(0,39)$&$(2,-11)$   
&6.56 
\\ \hline
% $[[156,4,16]]$ & $x^{-2}y$&$xy^{-2}$ &
% $(0,39)$&$(2,28)$   
% &6.56 
% \\ \hline

\rowcolor{green!40}$[[162,8,14]]$ & $\pmb{x^{-1}y^{-3}}$&$\pmb{x^3y^{-1}}$ &
$(0,9)$&$(9,-3)$   
&9.68 
\\ \hline

\rowcolor{mycolor7!70}$[[162,8,14]]$ & $\pmb{x^{-1}y^3}$&$\pmb{x^3y^{-1}}$ &
$(0,9)$&$(9,-3)$   
& ~9.68~ 
\\ \hline
% $[[162,8,14]]$ & $x^2y^3$&$x^{-3}y^2$ &
% $(0,9)$&$(9,6)$   
% &9.68 
% \\ \hline

\rowcolor{green!40}$[[\pmb{168,8,14}]]$ & $x^{-1}y^{-3}$&$x^3y^{-1}$ &
$(0,42)$&$(2,-16)$   
&9.33 
\\ \hline
% $[[168,8,14]]$ & $x^{-1}y^3$&$xy^3$ &
% $(0,42)$&$(2,30)$   
% &9.33 
% \\ \hline

%Added ~
\rowcolor{mycolor4!50}$[[\pmb{170,16,10}]]$ & $y^{-4}$&$x^4$ &
$(0,17)$&$(5,-7)$   
&9.41 
\\ \hline
% $[[170,16,10]]$ & $y^{-4}$&$x^4$ &
% $(0,17)$&$(5,10)$   
% &9.41 
% \\ \hline

$[[\pmb{174,4,18}]]$ & $x^{-8}y$&$x^6y^2$ &
$(0,3)$&$(29,1)$   
&7.45 
\\ \hline

\rowcolor{green!40}$[[180,8,16]]$ & $\pmb{x^{-1}y^{-3}}$&$\pmb{x^3y^{-1}}$ &
$(0,15)$&$(6,6)$   
&~11.38~ 
\\ \hline

\rowcolor{mycolor7!70}$[[180,8,16]]$ & $\pmb{x^{-1}y^3}$&$\pmb{x^3y^{-1}}$ &
$(0,15)$&$(6,3)$   
&~11.38~ 
\\ \hline

$[[\pmb{182,6,18}]]$ & $x^2y^3$&$x^4y$ &
$(0,7)$&$(13,1)$   
&10.68 
\\ \hline

\rowcolor{mycolor8!80}$[[\pmb{186,10,14}]]$ & $x^2y^3$&$x^2y^{-2}$ &
$(0,31)$&$(3,7)$   
&10.54
\\ \hline

\rowcolor{mycolor7!70}$[[\pmb{192,8,16}]]$ & $x^{-1}y^3$&$x^3y^{-1}$ &
$(0,12)$&$(8,2)$   &10.67 \\ \hline

%Added ~
\rowcolor{red!35}$[[\pmb{196,6,18}]]$ & $x^{-1}y^2$&$x^{-2}y^{-1}$ &
$~(0,49)~$&$~(2,-10)~$   &9.92 \\ \hline

\end{tabular}
\caption{Continuation of Table~\ref{tab: n_k_d 1} for $110 < n \leq 196$.}
\label{tab: n_k_d 2}
\end{table}

%%%%%%%%%%%%%%%%%%%%%%%%%%%%%%%%%%%%%%%%%%%%%%%%%%%%%%%%%%%
\renewcommand{\arraystretch}{1.2}
\setlength{\tabcolsep}{0pt} % Reduce column separation
\begin{table}[t]
\centering
\definecolor{mycolor1}{RGB}{255, 200, 100}  
\definecolor{mycolor2}{RGB}{200, 100, 200}
\definecolor{mycolor3}{RGB}{100, 180, 150}
\definecolor{mycolor4}{RGB}{100, 100, 150}
\definecolor{mycolor5}{RGB}{174, 217, 69}
\definecolor{mycolor6}{RGB}{250, 50, 200}
\definecolor{mycolor7}{RGB}{50, 250, 250}
\begin{tabular}{|c|c|c|c|c|c|}
\hline
$[[n,k,d]]$ & \makecell{$~f(x,y)=~$ \\ $1+x+...$}
& \makecell{$~g(x,y)=~$ \\ $1+y+...$} & $\vec{a}_1$   &$\vec{a}_2$   
&$\frac{kd^2}{n}$ 
\\ \hline

\rowcolor{mycolor5!40}$[[\pmb{198,8,16}]]$ & $x^{-4}$&$x^{-3}y^2$ &
$(0,33)$&$(3,9)$   &10.34 \\ \hline

$[[\pmb{204,4,20}]]$ & $x^{-3}y$&$x^{-1}y^{-2}$ &
$(0,51)$&$(2,14)$   &7.84 \\ \hline

\rowcolor{mycolor6!40} ~$[[\pmb{210,10,16}]]$~ & $x^{-3}y^2$&$x^{-3}y^{-1}$ &
$(0,21)$&$(5,10)$   &\pmb{12.19} \\ \hline

$[[\pmb{216,8,18}]]$ & $x^{-2}y^{-5}$&$x^{-1}y^{-3}$ &
$(0,54)$&$(2,16)$   &12 \\ \hline

$[[\pmb{222,4,20}]]$ & $x^{-6}y^{-1}$&$x^5$ &
$(0,3)$&$(37,2)$   &7.21 \\ \hline

\rowcolor{mycolor6!40}$[[\pmb{224,6,20}]]$ & $x^{-3}y^2$&$x^{-3}y^{-1}$ &
$(0,28)$&$(4,-6)$   &10.71 \\ \hline

\rowcolor{brown!40}$[[\pmb{228,4,20}]]$ & $x^{-2}y$&$xy^{-2}$ &
$(0,57)$&$(2,10)$   &7.02 \\ \hline

\rowcolor{green!40}$[[\pmb{234,8,18}]]$ & $x^{-1}y^{-3}$&$x^3y^{-1}$ &
$(0,39)$&$(3,-9)$   &11.08 \\ \hline

\rowcolor{mycolor7!70}$[[\pmb{234,8,18}]]$ & $x^{-1}y^3$&$x^3y^{-1}$ &
$(0,39)$&$(3,6)$   &11.08 \\ \hline

\rowcolor{mycolor5!40}$[[\pmb{238,6,20}]]$ & $x^{-4}$&$x^{-3}y^2$ &
$(0,7)$&$(17,1)$   &10.08 \\ \hline

\rowcolor{blue!30}$[[\pmb{240,8,18}]]$ & $x^{-2}y$&$xy^2$ &
$(0,10)$&$(12,3)$   &10.8 \\ \hline

%$[[\pmb{246,4,20}]]$ & $x^{-6}y^{-1}$&$x^4y$ &
%$(0,3)$&$(41,1)$   &6.50 \\ \hline
$[[\pmb{246,4,\leq22}]]$ & $x^3y$&$x^2y^{-2}$ &
~$(0,123)$~ &$(1,22)$   &7.87 \\ \hline

$[[\pmb{248,10,18}]]$ & $x^{-2}y$&$x^{-3}y^{-2}$ &
$(0,62)$&$(2,25)$   &\pmb{13.06} \\ \hline

$[[\pmb{252,12,16}]]$ & $x^{-3}y^{-1}$&$x^2y^{-2}$ &
$(0,18)$&$(7,7)$   &12.19 \\ \hline

$[[\pmb{254,14,16}]]$ & $x^{-1}y^{-3}$&$y^{-6}$ &
$(0,127)$&$(1,25)$   &\pmb{14.11} \\ \hline

$[[\pmb{258,4,\leq22}]]$ & $x^{-8}y^{-1}$&$x^5y$ &
$(0,3)$&$(43,1)$   &7.50 \\ \hline

$[[\pmb{264,8,20}]]$ & $xy^{-5}$&$xy^4$ &
$(0,66)$&$(2,28)$   &12.12 \\ \hline

$[[\pmb{266,6,\leq22}]]$ & $x^{-1}y^{-1}$&$x^5$ &
$(0,7)$&$(19,2)$   &10.92 \\ \hline

\rowcolor{green!40}$[[\pmb{270,8,20}]]$ & $x^{-1}y^{-3}$&$x^3y^{-1}$ &
$(0,15)$&$(9,6)$   &11.85 \\ \hline

\rowcolor{mycolor7!70}$[[\pmb{270,8,20}]]$ & $x^{-1}y^3$&$x^3y^{-1}$ &
$(0,45)$& ~$(3,-12)$~   &11.85 \\ \hline

\rowcolor{mycolor2!80}$[[\pmb{276,4,\leq24}]]$ & $x^{-3}y$&$x^3y^2$&
$(0,6)$&$(23,5)$   &8.35 \\ \hline

~$[[\pmb{280,6,\leq22}]]$~ & $xy^3$&$x^2y^{-2}$ &
$(0,28)$&$(5,12)$   & ~10.37~ \\ \hline

\rowcolor{mycolor7!70}$[[\pmb{282,4,\leq24}]]$ & $x^{-1}y^3$&$x^3y^{-1}$ &
$(0,141)$&$(1,7)$   &8.17 \\ \hline

\rowcolor{mycolor7!70}$[[\pmb{288,16,12}]]$ & $x^{-1}y^3$&$x^3y^{-1}$ &
$(0,12)$&$(12,0)$   &8 \\ \hline

\rowcolor{green!40} $~[[288,12,18]]~$ & $\pmb{x^{-1}y^{-3}}$&$\pmb{x^3y^{-1}}$ &
$(0,12)$&$(12,0)$   &\pmb{13.5} \\ \hline

$[[\pmb{292,18,8}]]$ & $y^2$&$x^{-4}y$ &
$(0,73)$&$(2,32)$   &3.95 \\ \hline

\end{tabular}
\caption{ Continuation of Table~\ref{tab: n_k_d 2} for $196 < n \leq 292$.
{ \change For $d>20$, the integer‐programming method fails to terminate in a reasonable time.
Accordingly, we use the probabilistic algorithm of Ref.~\cite{Pryadko2022GAP}, using between 5,000 and 10,000 information sets to estimate an upper bound on $d$.
Each computation is repeated several hundred times to ensure consistency, yielding a bound that we consider tight.
}
}
\label{tab: n_k_d 3}
\end{table}

%%%%%%%%%%%%%%%%%%%%%%%%%%%%%%%%%%%%%%%%%%%%%%%%%%%%%%%%%%%
\renewcommand{\arraystretch}{1.2}
\setlength{\tabcolsep}{0pt} % Reduce column separation
\begin{table}[t]
\centering
\definecolor{mycolor1}{RGB}{255, 200, 100}  
\definecolor{mycolor2}{RGB}{200, 100, 200}
\definecolor{mycolor3}{RGB}{100, 180, 150}
\definecolor{mycolor4}{RGB}{100, 100, 150}
\definecolor{mycolor5}{RGB}{174, 217, 69}
\definecolor{mycolor6}{RGB}{250, 50, 200}
\definecolor{mycolor7}{RGB}{50, 250, 250}
\definecolor{mycolor8}{RGB}{250, 250, 50}
\begin{tabular}{|c|c|c|c|c|c|}
\hline
$[[n,k,d]]$ & \makecell{$~f(x,y)=~$ \\ $1+x+...$}
& \makecell{$~g(x,y)=~$ \\ $1+y+...$} & $\vec{a}_1$   &$\vec{a}_2$   
&$\frac{kd^2}{n}$ 
\\ \hline

$[[\pmb{294,10,20}]]$ & $x^{-3}y$&$xy^{-3}$ &
$(0,21)$&$(7,7)$   &\pmb{13.61} \\ \hline

%Added ~
$[[\pmb{300,8,\leq22}]]$ & $x^{-1}y^{-4}$&$x^{-3}y^3$ &
$(0,75)$&$(2,26)$   &~12.91~ \\ \hline

%Added ~
\rowcolor{green!40} $[[\pmb{306,8,\leq22}]]$ &$x^{-1}y^{-3}$&$x^3y^{-1}$&
$(0,51)$&$(3,21)$   &12.65 \\ \hline

$[[\pmb{308,6,\leq24}]]$ & $x^{-1}y^{-2}$&$x^2y^{-1}$ &
$(0,77)$& ~$(2,-13)$~ & ~11.22~ \\ \hline

%Added ~
$[[\pmb{310,10,\leq22}]]$ & $x^3y^2$&$x^{-4}y^4$ &
$(0,31)$&$(5,11)$   &\pmb{15.61} \\ \hline

$[[\pmb{312,8,\leq22}]]$ & $x^{-1}y^3$&$xy^3$ &
$(0,78)$&$(2,-16)$   &~12.41~ \\ \hline

$[[\pmb{318,4,\leq26}]]$ & $x^3y^{-4}$&$x^{-1}y^{-3}$ &
~$(0,159)$~ &$(1,17)$   &8.50 \\ \hline
%$[[\pmb{318,4,\leq 24}]]$ & $x^{-6}y^{-1}$&$x^{-10}$ &
%$(0,3)$&$(53,1)$   &~7.25~ \\ \hline

$[[\pmb{322,6,\leq24}]]$ & $x^{-3}y^2$&$x^{-4}y^{-1}$ &
$(0,7)$&$(23,3)$   &10.73 \\ \hline

\rowcolor{green!40}$[[\pmb{324,8,\leq22}]]$  &$x^{-1}y^{-3}$&$x^3y^{-1}$&
$(0,18)$&$(9,6)$   &11.95\\ \hline

$[[\pmb{330,8,\leq24}]]$  &$x^{-6}y^2$&$x^2y^5$&
$(0,55)$&$(3,23)$   &13.96\\ \hline

$[[\pmb{336,10,\leq22}]]$  &$x^{-4}$&$x^{-1}y^{-3}$&
$(0,84)$&$(2,37)$   &14.40\\ \hline

\rowcolor{mycolor4!50}$[[\pmb{340,16,18}]]$  &$y^{-4}$&$x^4$&
$(0,34)$&$(5,-7)$   &\pmb{15.25}\\ \hline

\rowcolor{green!40}$[[\pmb{342,8,\leq22}]]$  &$x^{-1}y^{-3}$&$x^3y^{-1}$&
$(0,57)$&$(3,15)$   &11.32\\ \hline

\rowcolor{mycolor3!30}$[[\pmb{348,4,\leq26}]]$  &$x^{-2}y^2$&$x^{-1}y^{-2}$&
$(0,87)$&$(2,14)$   &7.77 \\ \hline

$[[\pmb{350,6,\leq26}]]$  &$x^2y^2$&$x^{-4}y$&
$(0,35)$&$(5,13)$   &11.58\\ \hline

\rowcolor{mycolor3!30}$[[\pmb{354,4,\leq28}]]$  &$x^{-2}y^2$&$x^{-1}y^{-2}$&
$(0,177)$&$(1,-53)$   &8.86\\ \hline

\rowcolor{mycolor7!70}~$[[360,12,\leq24]]$~  &$\pmb{x^{-1}y^3}$&$\pmb{x^3y^{-1}}$&
$(0,30)$&$(6,6)$   &\pmb{19.2}\\ \hline

$[[\pmb{364,6,\leq26}]]$  &$x^{-1}y^3$&$x^3$&
$(0,14)$&$(13,4)$   &11.14\\ \hline

\rowcolor{mycolor8!80}$[[\pmb{366,4,\leq28}]]$  &$x^2y^3$&$x^2y^{-2}$&
$(0,183)$&$(1,76)$   &8.57\\ \hline

$[[\pmb{372,10,\leq24}]]$  &$x^{-3}y^{-2}$&$x^{-1}y^{-3}$&
$(0,93)$&$(2,-16)$   &15.48\\ \hline

$[[\pmb{378,12,\leq22}]]$  &$x^3y^{-3}$&$x^4$&
$(0,21)$&$(9,6)$   &15.37\\ \hline

$[[\pmb{384,12,\leq24}]]$  &$x^{-4}y^{-3}$&$x^3y^{-1}$&
$(0,48)$&$(4,20)$   &18\\ \hline

$[[\pmb{390,8,\leq26}]]$  &$x^{-2}y^3$&$x^2y^3$&
$(0,15)$&$(13,1)$   &13.87\\ \hline

\rowcolor{mycolor6!40}$[[\pmb{392,6,\leq28}]]$  &$x^{-3}y^2$&$x^{-3}y^{-1}$&
$(0,28)$&$(7,7)$   &12\\ \hline

\rowcolor{green!40}$[[\pmb{396 ,8, \leq26}]]$  & $x^{-1} y^{-3}$ & $x^{3} y^{-1}$ &
$(0,66)$&$(3,18)$   & 13.66 \\ \hline

\end{tabular}
\caption{Continuation of Table~\ref{tab: n_k_d 3} for $292 < n \leq 400$.}
\label{tab: n_k_d 4}
\end{table}

We consider a translation-invariant $\mathbb{Z}_2$ CSS code whose stabilizers are expressed:
\begin{equation}
    A_v = 
    \begin{bmatrix}
        f(x,y) \\
        \rule{0pt}{1.1em}g(x,y) \\
        \hline
        0 \\
        0
    \end{bmatrix}, 
    \quad
    B_p = 
    \begin{bmatrix}
        0 \\
        0 \\
        \hline
        \rule{0pt}{1.1em}\overline{g(x,y)} \\
        \overline{f(x,y)}
    \end{bmatrix},
    \label{eq: stabilizer}
\end{equation}
{\change with $f(x,y), ~g(x,y) \in R$and $\overline{(\cdots)}$ denoting the {\bf antipode map}, a linear involution on $R$:
\begin{equation}
    x^n y^m \rightarrow \overline{x^n y^m}:= x^{-n} y^{-m}.
\end{equation}
$A_v$ and $B_p$ denote the $X$ and $Z$ stabilizer generators, respectively, forming the stabilizer group $\mathcal{S}$.\footnote{More precisely, $\mathcal{S}$ is generated by all lattice translations of $A_v$ and $B_p$, i.e.\ the set of operators
\[
\{~x^n y^m A_v,\; x^n y^m B_p \mid n,m\in\mathbb{Z}\}.
\]}
A simple example is the Kitaev toric code~\cite{bravyi1998quantum}, for which $f(x,y)=1+x$ and $g(x,y)=1+y$, and whose stabilizers are illustrated in Fig.~\ref{fig: polynomial convention}.
}
Next, we further introduce the {\bf generalized toric code}, a specific type of the bivariate bicycle code, defined by
\begin{eqs}
    f(x,y) &= 1 + x + x^a y^b, \\
    g(x,y) &= 1 + y + x^c y^d,
\label{eq: a b c d generalized toric code}
\end{eqs}
as illustrated in Fig.~\ref{fig: generalized_toric_code}. We refer to the stabilizer given by Eq.~\eqref{eq: a b c d generalized toric code} as the $(a,b,c,d)$-generalized toric code.

Next, we determine the {\bf excitation map} (error syndromes) for the single Pauli error at edges in the unit cell:
\begin{equation}
    \begin{aligned}
        \epsilon(\mathcal{X}_1) &= [A_v \cdot \mathcal{X}_{1} ~,~ B_p \cdot \mathcal{X}_{1}] 
        = [0,\; g(x,y)], \\
        \epsilon(\mathcal{X}_2) &= [A_v \cdot \mathcal{X}_{2} ~,~ B_p \cdot \mathcal{X}_{2}] 
        = [0,\; f(x,y)], \\
        \epsilon(\mathcal{Z}_1) &= [A_v \cdot \mathcal{Z}_{1} ~,~ B_p \cdot \mathcal{Z}_{1}] 
        = [\overline{f(x,y)},\; 0], \\
        \epsilon(\mathcal{Z}_2) &= [A_v \cdot \mathcal{Z}_{2} ~,~ B_p \cdot \mathcal{Z}_{2}] 
        = [\overline{g(x,y)},\; 0],
    \end{aligned}
    \label{eq: get_error_syndromes}
\end{equation}
{\change
where $\cdot$ is the symplectic product defined by
\begin{eqs}
    v_1 \cdot v_2 = \overline{v}_1^{T} \Lambda v_2,
\end{eqs}
where $T$ denotes matrix transpose and
\begin{eqs}
    \Lambda=
    \left[\begin{array}{cc | cc}
        0 & 0 & 1 & 0 \\
        0 & 0 & 0 & 1 \\
        \hline
        -1 & 0 & 0 & 0 \\
        0 & -1 & 0 & 0 \\
    \end{array}\right]
\end{eqs}
is the standard symplectic form in this basis.\footnote{Since in $\mathbb{Z}_2$, $-1$ is identified wtih $+1$, the minus sign is immaterial; it is included solely to preserve consistency with the $\mathbb{Z}_d$ qudit generalization.}
The symplectic product captures the commutation relations: two Pauli operators anticommute precisely when their symplectic product has the nonzero constant term.
See Ref.~\cite{liang2023extracting} for further details.
}

To verify that the stabilizer Hamiltonian satisfies the topological order (TO) condition, we check that any local operator commuting with the stabilizers can be expressed as a finite product of stabilizers~\cite{haah_module_13}:
\begin{equation}
    \ker \epsilon = \mathcal{S}.
\end{equation}
In reference to the polynomials $f(x, y)$ and $g(x, y)$, the condition can be reformulated as~\cite{eberhardt2024logical}:
\begin{equation}
    \langle f(x, y) \rangle \cap \langle g(x, y) \rangle = \langle f(x, y) g(x, y) \rangle,
\label{eq: TO condition for f and g}
\end{equation}
where $\langle p(x, y) \rangle$ denotes the ideal in $R$ generated by the polynomial $p(x, y)$. This implies that the polynomials $f(x, y)$ and $g(x, y)$ are coprime.
The algorithm for verifying whether $A_v$ and $B_p$ in Eq.~\eqref{eq: stabilizer} satisfy the TO condition is provided in Ref.~\cite{liang2023extracting}.

%%%%%%%%%%%%%%%%%%%%%%%%%%%%%%%%%%%%%%%%%%%%%%%%%%%%%%%%%%%%%%%%%%%%%%%%%%%%%%%%
\subsection{Classification of anyons on an infinite plane}
Anyons are defined as violations of stabilizers, i.e., the mapping \cite{liang2024operator}:
\begin{equation}
    \phi:\mathcal{S} \rightarrow\ZZ_2.
\label{eq: stabilizer violation}
\end{equation}
We first focus on $m$-type anyons, which correspond to violations of $B_p$ caused by Pauli $X$ operators. These anyons take the form
\begin{equation}
    v_m = [0,a(x,y)], \quad a(x,y) \in R,
\end{equation}
where $a(x,y)$ records the location where $B_p$ is violated.
From Eq.~\eqref{eq: get_error_syndromes}, both $[0, g(x,y)]$ and $[0,f(x,y)]$ represent trivial anyons since local operators create them. Thus, the nontrivial $m$-type anyons are classified by the quotient ring:
\begin{equation}
    \text{$m$-type anyons} = \frac{\mathbb{Z}_2[x,y,x^{-1},y^{-1}]}{\langle f(x,y),~g(x,y) \rangle}=:\frac{R}{I},
\label{eq: m anyon}
\end{equation}
which corresponds to the Laurent polynomial ring $R$ modulo the ideal $I$ generated by $f(x,y)$ and $g(x,y)$. 
Similarly, we have
\begin{equation}
    \text{$e$-type anyons} = \frac{\mathbb{Z}_2[x,y,x^{-1},y^{-1}]}{\langle \overline{f(x,y)},~\overline{g(x,y)} \rangle}.
\label{eq: e anyon}
\end{equation}
The numbers of $e$-anyons and $m$-anyons are equal since they can be re-arranged and paired as $\{e_1, m_1\}, \{e_2, m_2\}, \dots$, forming a direct sum of Kitaev toric codes~\cite{bombin_Stabilizer_14, ruba2024homological}. By Eq.~\eqref{eq: gsd and anyon number}, this leads to the following theorem:
\begin{theorem}
    The maximal logical dimension $k_{\mathrm{max}}$ of the stabilizer codes in Eq.~\eqref{eq: stabilizer}, parameterized by two polynomials $f(x,y)$ and $g(x,y)$ on a torus, is given by twice the number of independent monomials in $R$ quotient by the ideal $I=\langle f(x, y), g(x,y) \rangle$:
    \begin{equation}
        k_\mathrm{max} = 2 \dim \left(
        \frac{\mathbb{Z}_2[x,y,x^{-1},y^{-1}]}{\langle f(x,y),~g(x,y) \rangle}
        \right).
    \label{eq: maximal logical dimension}
    \end{equation}
\label{thm: maximal logical dimension}
\end{theorem}
{\change
Indeed, Corollary 4.5 of Ref.~\cite{haah_module_13} asserts that
\begin{equation}
  k = \dim_{\mathbb{F}_2}\bigl(\mathrm{coker}\,\epsilon\bigr)
\end{equation}
for an untwisted torus of arbitrary size. Since it is upper bounded by Eq.~\eqref{eq: maximal logical dimension}, Theorem~\ref{thm: maximal logical dimension} follows immediately.
Here, however, we prefer to invoke the anyon picture both to build physical intuition and to streamline our subsequent discussion of the associated logical operators.
}

The ground state degeneracy could depend on the torus length, as observed in the Wen plaquette model~\cite{wen2003quantum}, the Watanabe-Cheng-Fuji toric code~\cite{watanabe2023ground}, and fraction models~\cite{chamon2005quantum, haah2011local, vijay2016fracton, shirley2018fracton, dua2019sorting, nandkishore2019fractons, Dua2019Compactifyingfracton,Li2020Fracton,  pretko2020fracton, qi2021fracton, Li2021Fracton, Song2022Optimal, tan2023fracton, Li2023Hierarchy, Zhu2023Fracton, Li2024fractonToeplitz, Song2024FractonStatistics, Zhang2024HigherCA}.
Eq.~\eqref{eq: gsd and anyon number} holds in the infrared limit, giving the maximal Hilbert space dimension for a finite torus. We will examine finite-size effects in more detail later.

\begin{figure*}[t]
    \centering
    \subfigure[\normalsize Twisted torus in three-dimensional space.]{\raisebox{4ex}{\includegraphics[width=0.48\linewidth]{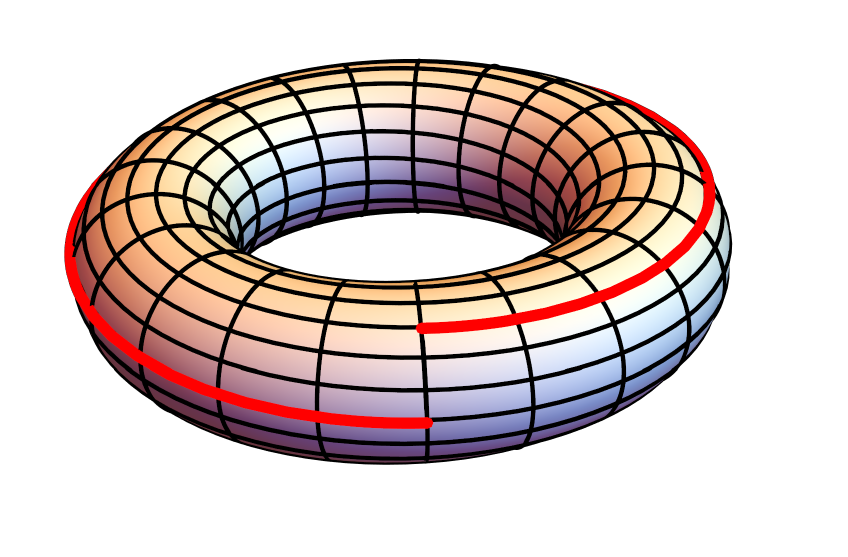}}}
    \subfigure[\normalsize Two-dimensional projection of the twisted torus.]{\includegraphics[width=0.35\linewidth]{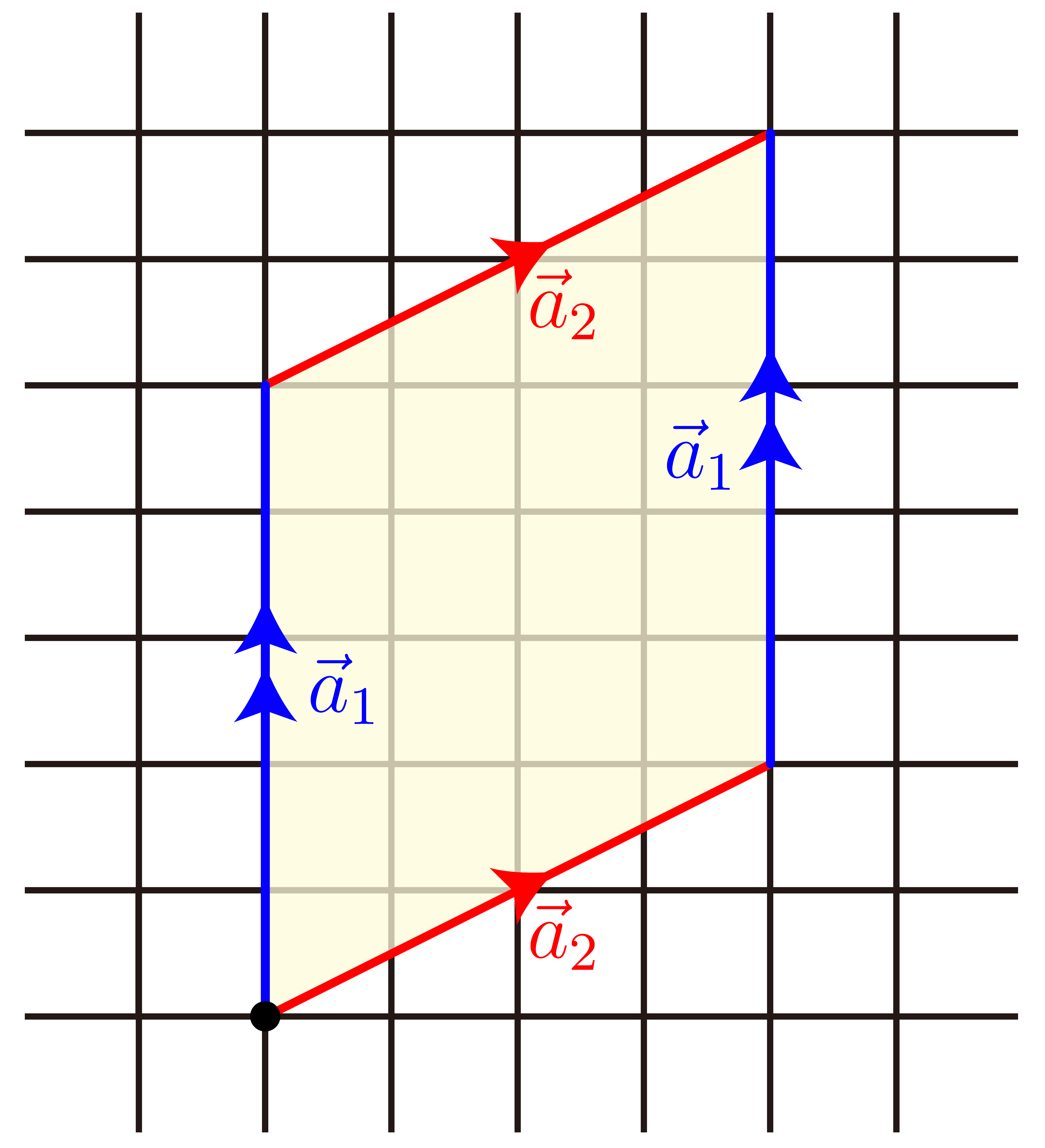}}
    \caption{(a) A twisted torus embedded in three-dimensional space. The torus undergoes a twist along its longitudinal direction by an angle that is a fraction of $2\pi$, as indicated by the red curve tracing the large cycle. (b) A two-dimensional projection of the twisted torus, where points related by the lattice vectors $\vec{a}_1$ and $\vec{a}_2$ are identified. The parallelogram's opposite edges are identified, forming the twisted torus. The twisted torus was previously employed in Ref.~\cite{Hastings2021Fiber} to construct fiber bundle codes.}
    \label{fig: twisted torus}
\end{figure*}

Now, we present several examples to illustrate the application of this theorem.
\begin{example}
{\bf Kitaev toric code.} The Kitaev toric code corresponds to $f(x, y) = 1 + x$ and $g(x, y) = 1 + y$ (as shown in Fig.~\ref{fig: generalized_toric_code} without extra terms). It has only one independent monomial, $1$, i.e., all monomial $x^a y^b$ can be expressed as
\begin{equation}
    x^a y^b = a_1 + p(x,y) f(x,y) + q(x,y) g(x,y),
\end{equation}
with $a_1 \in \ZZ_2$ and $p(x,y),~q(x,y) \in R$. For example,
\begin{equation}
    x^2 = 1 + (1+x)(1+x).
\end{equation}
This implies $k_{\mathrm{max}}=2$. This agrees that the Kitaev toric code on a torus has $2^{k_{\mathrm{max}}} = 4$ ground states.
\label{example: Kitaev toric code}
\end{example}

\begin{example}
{\bf Color code.} The color code is defined by $f(x, y) = 1 + x + xy$ and $g(x, y) = 1 + y + xy$~\cite{bombin2006topological, eberhardt2024pruning}. It is straightforward to verify that the independent monomials are $1$ and $x$, and all other monomials can be generated as
\begin{equation}
    x^a y^b = a_1 + a_x x+ p(x,y) f(x,y) + q(x,y) g(x,y),
\end{equation}
with $a_1,~a_x \in \ZZ_2$ and $p(x,y),~q(x,y) \in R$. Thus, the color code has $k_{\mathrm{max}} = 4$, consistent with the fact that it can be viewed as a folded toric code, forming a direct sum of two Kitaev toric codes~\cite{kubica2015unfolding}.
\end{example}
In the examples above, the polynomials $f$ and $g$ are simple enough to check independent monomials by hand. However, determining independent monomials for general $f$ and $g$ might not be obvious. Therefore, we will introduce a systematic way below using the \textbf{Gr\"obner basis}.

%%%%%%%%%%%%%%%%%%%%%%%%%%%%%%%%
\begin{example}
{\bf $(-1,3,3,-1)$-generalized toric code.}  
This corresponds to the stabilizers of the gross code in Ref.~\cite{Bravyi2024HighThreshold} and the $(3,3)$-BB code in Ref.~\cite{liang2024operator}, described by the following polynomials:
\begin{eqs}
    f(x, y) &= 1 + x + x^{-1} y^3, \\
    g(x, y) &= 1 + y + x^3 y^{-1},
\end{eqs}
or equivalently,
\begin{eqs}
    f'(x, y) &= x + x^2 + y^3, \\
    g'(x, y) &= y + y^2 + x^3.
\end{eqs}
We now compute the Gr\"obner basis with lexicographic ordering $x < y$ using Buchberger’s algorithm \cite{Buchberger1965algorithm}:\footnote{A useful technique is to compute the Gr\"obner basis for the polynomials $f'(x, y) = x f(x, y)$ and $g'(x, y) = y g(x, y)$ to ensure all exponents remain non-negative. Alternatively, one can introduce new variables $\overline{x}$ and $\overline{y}$ to represent $x^{-1}$ and $y^{-1}$, respectively, and include the extra polynomials $x\overline{x} - 1$ and $y\overline{y} - 1$ in the Gr\"obner basis computation.}
\begin{eqs}
    h(x,y) &= 1 + y + y^3 + y^5 + y^6, \\
    i(x,y) &= x + x y + x y^2 + y^3 + y^6, \\
    j(x,y) &= x + x^2 + y^3,
\label{eq: (-1,3,3,-1) h i j}
\end{eqs}
such that
\begin{equation}
    \langle f'(x,y),~g'(x,y) \rangle = \langle h(x,y), ~i(x,y), ~j(x,y) \rangle.
\end{equation}
We can think of the Gr\"obner basis as a generalized version of the Gaussian elimination process, where the polynomials are reorganized into a standard form (analogous to row echelon form). This procedure allows us to identify a set of simplified polynomials that span the same linear space.
From the Gr\"obner basis, we identify the following 8 independent monomials:
\begin{equation}
    1, y, y^2, y^3, y^4, y^5, x, xy,
\end{equation}
which implies that $k_{\mathrm{max}} = 16$.  
These independent monomials can be found algorithmically by first listing all possible monomials in the range specified by Eq.~\eqref{eq: (-1,3,3,-1) h i j}:
\begin{equation}
    \{x^a y^b \mid 0 \leq a < 2, ~0 \leq b < 6\},
\label{eq: range of monomials}
\end{equation}
and then applying Gaussian elimination to determine their linear relations, reducing the set accordingly \cite{liang2023extracting}.

Note that this code has $k_{\mathrm{max}} = 16$, which is larger than the $k = 12$ in the gross code $[144, 12, 12]$ (Ref.~\cite{Bravyi2024HighThreshold}). This discrepancy arises from the effect of finite torus, which will be further discussed in the next section.
\label{example: -1 3 3 -1 generalized TC on infinite plane}
\end{example}

%%%%%%%%%%%%%%%%%%%%%%%%%%%%%%%%

We observe that the $(-1,3,3,-1)$-generalized toric code appears multiple times in Tables~\ref{tab: n_k_d 1}, \ref{tab: n_k_d 2}, ~\ref{tab: n_k_d 3}, and~\ref{tab: n_k_d 4}, highlighted in blue. Additionally, we introduce another set of stabilizers, the $(-1,-3,3,-1)$-generalized toric code, which also appear frequently in these tables, highlighted in green.

%%%%%%%%%%%%%%%%%%%%%%%%%%%%%%%%
\begin{example}
{\bf $(-1,-3,3,-1)$-generalized toric code.}  
This corresponds to the stabilizers of the $(3,-3)$-BB code in Ref.~\cite{liang2024operator}, described by the following polynomials:\footnote{Alternatively, the polynomials can be expressed as
\begin{equation}
    f(x, y) = x + x^2 + y^{-3}, \quad g(x, y) = y + y^2 + x^3.
\end{equation}
}
\begin{eqs}
    f(x, y) &= 1 + x + x^{-1}y^{-3}, \\
    g(x, y) &= 1+ y + x^{3} y^{-1}.
\end{eqs}
We first compute its Gr\"obner basis:
\begin{eqs}
    h(x,y) &= 1 + y + y^3 + y^4 + y^6 + y^{10} + y^{11}, \\
    i(x,y) &= x + y + x y + y^2 + x y^2 + y^5 + y^{10}, \\
    j(x,y) &= x + x^2 + y^4 + y^5 + y^7 + y^{10},
\label{eq: (-1,-3,3,-1) h i j}
\end{eqs}
such that
\begin{equation}
    \langle f(x,y),~g(x,y) \rangle = \langle h(x,y), ~i(x,y), ~j(x,y) \rangle.
\end{equation}
From the Gr\"obner basis, it is straightforward to identify $13$ independent monomials:  
\begin{equation}
    1, y, y^2, y^3, y^4, y^5, y^6, y^7, y^8, y^9, y^{10}, x, xy,
\end{equation}
implying $k_{\mathrm{max}} = 26$.  
\label{example: -1 -3 3 -1 generalized TC on infinite plane}
\end{example}
%%%%%%%%%%%%%%%%%%%%%%%%%%%%%%%%
The fact that the polynomial $h(x,y)$ is alway univariate in terms of $y$ follows from the TO condition~\eqref{eq: TO condition for f and g}. By B\'ezout’s theorem~\cite{Cox2015Ideals}, two coprime polynomials in two variables intersect at finitely many points. Equivalently, the ideal $I = \langle f, g \rangle$ is zero-dimensional, so the quotient ring $\mathbb{Z}_2[x,y]/I$ is a finite-dimensional vector space, where each variable is algebraic. By the ``Shape Lemma,'' the lexicographic Gr\"obner basis of $I$ must contain a univariate polynomial~\cite{Becker1994ShapeLemma, Szanto2024LectureNotes}.

\subsection{Effects of finite geometries on tori}
Previously, we considered the infinite plane geometry when deriving the anyons in Eq.~\eqref{eq: m anyon}. However, the period of an anyon, defined as the shortest translation distance it can move while preserving its syndrome pattern~\eqref{eq: stabilizer violation}, is often greater than $1$ \cite{haah_module_13, watanabe2023ground, liang2023extracting}. This imposes a compatibility condition when compactifying the infinite plane onto a finite torus. If the period of an anyon does not divide the torus length, the anyon vanishes on the finite torus.
The following theorem specifies the condition on the lengths of the torus for the stabilizer codes to have maximal {\change logical dimension} $k=k_\mathrm{max}$:
\begin{theorem}
    Consider the stabilizer code~\eqref{eq: stabilizer} implemented on an untwisted $L_x \times L_y$ torus. The minimal torus lengths $L_x$ and $L_y$ for which the stabilizer code achieves $k = k_\mathrm{max}$ are the smallest values satisfying  
    \begin{equation}
        x^{L_x} - 1, \quad y^{L_y} - 1 \in \langle f(x,y),~g(x,y) \rangle.
        \label{eq: x Lx -1 y Ly -1 in ideal}
    \end{equation}
\label{thm: full logical dimension on untwisted torus}
\end{theorem}
\begin{proof}
     Consider an $m$-anyon represented by the equivalence class $[0, a(x,y)]$. We must ensure that moving it along a large cycle in the $y$-direction does not change the superselection sector, so the ground state space remains invariant. This requires that after a full translation by $L_y$, the anyon remains in the same equivalence class:
    \begin{equation}
        y^{L_y} a(x,y) = a(x,y) + p(x,y) f(x,y) + q(x,y) g(x,y),
    \end{equation}
    for some polynomials $p(x,y)$ and $q(x,y)$.  
    Since this equation must hold for arbitrary $a(x,y)$, the term $y^{L_y} - 1$ must be expressible as a linear combination of $f(x,y)$ and $g(x,y)$, meaning:
    \begin{equation}
        y^{L_y} - 1 \in \langle f(x,y),~g(x,y) \rangle.
        \label{eq: y Ly -1 in ideal}
    \end{equation}
    The same argument applies to $e$-anyons and translations in the $x$-direction, leading to the analogous condition:
    \begin{equation}
        x^{L_x} - 1 \in \langle f(x,y),~g(x,y) \rangle.
        \label{eq: x Lx -1 in ideal}
    \end{equation}
    Therefore, the minimal values of $L_x$ and $L_y$ satisfying these conditions determine the torus dimensions required to achieve the maximal logical dimension $k = k_{\mathrm{max}}$.
\end{proof}
The minimal values of $L_x$ and $L_y$ satisfying Eq.~\eqref{eq: x Lx -1 y Ly -1 in ideal} can be computed as follows:
\begin{corollary}
    Given the polynomials $f(x,y)$ and $g(x,y)$, we obtain univariate polynomials $h(y)$ and $h'(x)$ by computing their Gr\"obner bases with two different orderings. The values $L_x$ and $L_y$ are determined by the minimal solutions of the following divisibility conditions:
    \begin{equation}
        h'(x) \mid x^{L_x} - 1, \quad h(y) \mid y^{L_y} - 1.
    \end{equation}
    If $h(y)$ is an irreducible polynomial, then $L_y$ must be a divisor of $2^{\deg(h(y))} - 1$~\cite{lidl1986introduction}, where $\deg(p(y))$ denotes the highest degree of the polynomial $p(y)$. Therefore, it suffices to check the factors of $2^{\deg(h(y))} - 1$. If $h(y)$ is reducible, we decompose it into irreducible factors and determine their shortest periodicities separately. The same procedure applies to $h'(x)$.
\end{corollary}

For instance, in Example~\ref{example: -1 3 3 -1 generalized TC on infinite plane}, the $(3,3)$-BB code, its Gr\"obner basis includes the polynomial
\begin{equation}
    h(y) = 1 + y + y^3 + y^5 + y^6.
\end{equation}
Verified by a computer, we determine that the minimal $L_y$ satisfying  
\begin{equation}
    h(y) \mid y^{L_y} - 1
\end{equation}
is $L_y = 12$. In this example, the stabilizers exhibit symmetry under the exchange of $x$ and $y$, implying that the period in the $x$-direction is also $L_x = 12$. Therefore, when the stabilizers are placed on a $12 \times 12$ torus, the logical dimension is $k = k_{\mathrm{max}} = 16$.

Similarly, in Example~\ref{example: -1 -3 3 -1 generalized TC on infinite plane}, the Gr\"obner basis includes the polynomial  
\begin{equation}
    h(y) = 1 + y + y^3 + y^4 + y^6 + y^{10} + y^{11}.
\end{equation}
Verified by a computer, we determine that the minimal $L_y$ with
\begin{equation}
    h(y) \mid y^{L_y} - 1
\end{equation}
is $L_y = 762$.  
Similarly, in the $x$-direction, an alternative Gr\"obner basis can be obtained by using a different monomial ordering:  
\begin{equation}
    h'(x) = 1 + x^4 + x^6 + x^7 + x^9 + x^{10} + x^{11},
\end{equation}
which satisfies $h'(x) \mid x^{762} - 1$ and yields $L_x = 762$.
Thus, to ensure the full logical dimension with $k = k_\mathrm{max} = 26$, the code should be implemented on an untwisted torus with size $762 \times 762$.

Theorem~\ref{thm: full logical dimension on untwisted torus} can be easily generalized to apply to twisted tori as well:
\begin{corollary}
    Consider the stabilizer code~\eqref{eq: stabilizer} defined on a twisted torus with two lattice vectors, $\vec{a}_1 = (0, \alpha)$ and $\vec{a}_2 = (\beta, \gamma)$, as illustrated in Fig.~\ref{fig: twisted torus}. The torus for which the stabilizer code achieves the maximum code dimension, $k = k_\mathrm{max}$, corresponds to the values of $\alpha$, $\beta$, and $\gamma$ that satisfy the following condition:
    \begin{equation}
        y^{\alpha} - 1, \quad x^\beta y^{\gamma} - 1 \in \langle f(x,y),~g(x,y) \rangle.
    \end{equation}
\end{corollary}

For Example~\ref{example: -1 -3 3 -1 generalized TC on infinite plane}, i.e., the $(3, -3)$-BB code, we can verify the following relation:
\begin{equation}
    x^{6}y^{360}-1, \quad y^{762}-1 \in \langle f(x,y),~g(x,y) \rangle,
\end{equation}
by computing the Gr\"obner basis of the ideal
\begin{equation}
    \langle f(x,y), ~ g(x,y) \rangle,
\end{equation}
and the Gr\"obner basis of the (extended) ideal
\begin{equation}
    \langle f(x,y), ~ g(x,y), ~ x^{6}y^{360}-1, ~ y^{762}-1 \rangle,
\end{equation}
and verify that they are identical.
In other words, the $(3, -3)$-BB code achieves its maximal logical dimension of $k = 26$ on the twisted torus with $\vec{a}_1 = (0, 762)$ and $\vec{a}_2 = (6, 360)$.
Compared to the untwisted $762 \times 762$ torus, the twisted torus achieves the same logical dimension with a significantly smaller system size.

{\change
We emphasize that parametrizing twisted tori by
\begin{equation}
\vec{a}_1=(0,\alpha),\quad \vec{a}_2=(\beta,\gamma)
\end{equation}
indeed covers every case. As an example, begin with a generic twisted torus with
\begin{equation}
    \vec{a}_1=(2,5),~ \vec{a}_2=(7,3) \Rightarrow M:=
    \begin{pmatrix}
    2 & 5 \\
    7 & 3
    \end{pmatrix}.
\end{equation}
Computing its Hermite normal form yields a unimodular transformation $U$ and an upper‐triangular matrix $H$:
\begin{equation}
U ~M
=\begin{pmatrix}
4 & -1 \\
7 & -2
\end{pmatrix}
\begin{pmatrix}
2 & 5 \\
7 & 3
\end{pmatrix}
=
\begin{pmatrix}
1 & 17 \\
0 & 29
\end{pmatrix}
= H,
\label{eq: 1 17 0 29 from 2 5 7 3}
\end{equation}
with
\begin{equation}
\bigl|\det U\bigr| = 1.
\end{equation}
Therefore, an equivalent choice of basis vectors is
\begin{equation}
\vec{a}_1=(0,29),\quad \vec{a}_2=(1,17),
\end{equation}
which represents the same twisted torus.

}

Let's consider another interesting example:
\begin{example}
{\bf $(-1,-4,4,-1)$-generalized toric code.}  
This is also known as the $(4,-4)$-BB code, whose stabilizers are defined by 
\begin{eqs}
    f(x, y) &= 1 + x + x^{-1} y^{-4}, \\
    g(x, y) &= 1 + y + x^{4} y^{-1}.
\end{eqs}
Using Gr\"obner basis computation, we obtain an alternative set of polynomials:  
\begin{eqs}
    h(x,y) &= 1 + y^{17} + y^{20}, \\
    i(x,y) &= x + y + y^2 + y^9 + y^{12} + y^{13} + y^{16},
\end{eqs}
generating the same ideal.
It follows that there are $20$ independent monomials:
\begin{equation}
    1, y, y^2, \dots, y^{19},
\end{equation}
implying $k_\mathrm{max} = 40$. Notably, $1 + y^{17} + y^{20}$ is a primitive polynomial over $\mathbb{Z}_2$, i.e., $\mathbb{Z}_2[y]/(1+y^{17}+y^{20})$ forms a finite field. Consequently, the minimal length $L_y$ satisfying  
\begin{equation}
    1 + y^{17} + y^{20} \mid y^{L_y} - 1
\end{equation}
is given by $L_y = 2^{20}-1 = 1,048,575$.  
Similarly, computing the Gr\"obner basis using an alternative ordering yields  
\begin{equation}
    h'(x,y) = 1 + x + x^2 + x^{18} + x^{20}.
\end{equation}
The minimal $L_x$ satisfying  
\begin{equation}
    1 + x + x^2 + x^{18} + x^{20} \mid x^{L_x}-1
\end{equation}
is determined as $L_x = (2^{20}-1)/15 = 69,905$.  
Therefore, for the $(4,-4)$-BB code to achieve the full logical dimension $k = k_\mathrm{max} = 40$ on an untwisted torus, its minimum size is $69,905 \times 1,048,575$.
\end{example}

{\change

% \subsubsection*{Logical operators for $k=k_\mathrm{max}$}
\begin{center}
\textit{Logical operators for $k=k_\mathrm{max}$}
\end{center}

When the logical dimension attains its maximum value $k_{\max}$, the logical $X$ and $Z$ operators admit a particularly simple description: they are the anyon string operators wrapping the two non-contractible loops of the (twisted) torus. We begin by constructing the logical $X$ operators, which correspond to the $m$-type anyon 
\begin{equation}
v_m = [0,~a(x,y)],
\end{equation}
the violation of the $B_p$ stabilizers. It suffices to solve the string operator for the fundamental anyon 
\begin{equation}
v_m^0 = [0,1],
\end{equation}
since all other $m$-type string operators follow by multiplying with appropriate Laurent monomials in $R$.
According to Ref.~\cite{liang2023extracting}, the {\bf anyon equation} along the $(0,\alpha)$ direction is
\begin{equation}
  (y^\alpha - 1)~[0,1]
    = p_1(x,y)~\epsilon(\mX_1)
    + p_2(x,y)~\epsilon(\mX_2),
\label{eq: m-string-anyon}
\end{equation}
and the corresponding string operator is
\begin{equation}
  O^{(0,\alpha)}_{\mathrm{string}}
  = \begin{bmatrix}
      p_1(x,y) \\
      p_2(x,y) \\ \hline
      0 \\
      0
    \end{bmatrix},
\end{equation}
which creates two anyons $[0,1]$ and $[0, y^\alpha]$. On the torus, these two positions coincide and the anyons annihilate, so $O^{(0,\alpha)}_{\mathrm{string}}$ is a logical $X$ operator.\footnote{We have omitted the $Z$-syndrome terms since they only excite $A_v$ stabilizers.}

Substituting the syndromes from Eq.~\eqref{eq: get_error_syndromes} into Eq.~\eqref{eq: m-string-anyon} yields
\begin{equation}
  y^\alpha - 1 = p_1(x,y)~g(x,y) + p_2(x,y)~f(x,y).
\end{equation}
By hypothesis, maximal logical dimension requires $h(y)\mid (y^\alpha-1)$, so
\begin{equation}
  y^\alpha - 1 = q(y)~h(y).
\label{eq: y alpha-1 = q h}
\end{equation}
Moreover, the Gr\"obner basis computation of the ideal $I=\langle f(x,y),~g(x,y)\rangle$ produces $h(y)$ as a combination of $f$ and $g$ by construction. Thus one immediately obtains $p_1$ and $p_2$ and hence the operator $O^{(0,\alpha)}_{\mathrm{string}}$.
Let the monomials
\begin{equation}
  x^{n_1}y^{m_1},~x^{n_2}y^{m_2},~\dots,~x^{n_{k/2}}y^{m_{k/2}}
\end{equation}
be the basis of the quotient module $R/I$. Then the logical $X$ operators along the $(0,\alpha)$ cycle are
\begin{equation}
  \overline{X}_j = x^{n_j}y^{m_j}~O^{(0,\alpha)}_{\mathrm{string}},
  \quad \forall j \in \{ 1,2,\dots,k/2\}.
\label{eq: logical X operators 1}
\end{equation}
Similarly, for the second non‐contractible cycle labeled by $(\beta,\gamma)$, we solve
\begin{equation}
  (x^\beta y^\gamma - 1)~[0,1]
    = p_3(x,y)~\epsilon(\mX_1)
    + p_4(x,y)~\epsilon(\mX_2),
\end{equation}
obtaining
\begin{equation}
  O^{(\beta,\gamma)}_{\mathrm{string}}
  = 
  \begin{bmatrix}
    p_3(x,y) \\
    p_4(x,y) \\ \hline
    0 \\
    0
  \end{bmatrix}.
\end{equation}
Hence, the remaining logical $X$ operators are
\begin{equation}
  \overline{X}_{j + k/2}
  = x^{n_j}y^{m_j}~O^{(\beta,\gamma)}_{\mathrm{string}},
  \quad \forall j \in \{ 1,2,\dots,k/2\}.
\label{eq: logical X operators 2}
\end{equation}
The logical $Z$ operators are obtained in complete analogy by interchanging $X$ and $Z$ syndromes.

There is a subtlety on the twisted torus ($\gamma\neq0$): the factor $x^\beta y^\gamma - 1$ is no longer a univariate polynomial, so Eq.~\eqref{eq: y alpha-1 = q h} cannot be applied directly. In this case, we introduce a new variable $\tilde{x}$ defined by
\begin{equation}
    x^\beta y^\gamma
    =\bigl(x^{\tfrac{\beta}{\gcd(\beta,\gamma)}}\,y^{\tfrac{\gamma}{\gcd(\beta,\gamma)}}\bigr)^{\gcd(\beta,\gamma)}
    =:\tilde{x}^{\gcd(\beta,\gamma)},
\end{equation}
where $\gcd(\beta,\gamma)$ denotes the greatest common divisor of $\beta$ and $\gamma$.  Since the reduced exponents
\begin{equation}
    \tilde{\beta} \;:=\;\frac{\beta}{\gcd(\beta,\gamma)},
    \quad
    \tilde{\gamma} \;:=\;\frac{\gamma}{\gcd(\beta,\gamma)}
\end{equation}
are coprime, Bézout’s identity guarantees integers $\delta,\eta$ such that
\begin{equation}
    \beta\,\delta + \gamma\,\eta = 1.
\end{equation}
Hence, we introduce the basis transformation on exponent vectors:
\begin{equation}
  \begin{pmatrix}a'\\b'\end{pmatrix}
  =
  \begin{pmatrix}\beta & -\eta \\ \gamma & \delta\end{pmatrix}
  \begin{pmatrix}a\\b\end{pmatrix},
\end{equation}
which induces the monomial substitution
\begin{equation}
  \tilde{x}^a \tilde{y}^b \longleftrightarrow x^{a'}y^{b'}.
\label{eq: basis transformation of x and y}
\end{equation}
Since
\begin{equation}
  \det\!\begin{pmatrix}\beta & -\eta \\ \gamma & \delta\end{pmatrix}
  = \beta\delta + \gamma\eta = 1,
\end{equation}
this transformation is invertible over $\mathbb{Z}$.  In these new coordinates,
\begin{equation}
  x^\beta y^\gamma - 1 = \tilde{x}^{\gcd(\beta,\gamma)} - 1,
\end{equation}
which is now a univariate polynomial, so the preceding analysis applies directly.

}

%\subsubsection*{Frustration between anyon periodicity and finite system size}\label{sec: Frustration}
\newpage

\begin{center}
\textit{Frustration between anyon periodicity and system size}
\end{center}
So far, we have discussed how to preserve the full ground state degeneracy $k_\mathrm{max}$ on a torus. However, in practical scenarios, when the torus length is not a multiple of all anyon periodicities, only a subset of the anyons survives, resulting in a ground state space with $k < k_\mathrm{max}$.

To analyze the effects of the periodic boundary conditions on an $L_x \times L_y$ torus, we impose the additional constraints on the polynomials:  
\begin{equation}
    x^{L_x}-1=0, \quad y^{L_y}-1=0.
\end{equation}
Consequently, the number of independent $m$-type anyons in Eq.~\eqref{eq: m anyon} is reduced to:  
\begin{equation}
    \dim \left(
    \frac{\mathbb{Z}_2[x,y,x^{-1},y^{-1}]}{\langle f(x,y),~g(x,y),~x^{L_x} - 1,~y^{L_y} - 1 \rangle}
    \right).
\label{eq: m anyon standard PBC}
\end{equation}
Moreover, we can consider the torus with twisted periodic boundary conditions, as shown in Fig.~\ref{fig: twisted torus}. The twisted torus can be specified by two vectors $\vec{a}_1 = (0, \alpha)$ and $\vec{a}_2 = (\beta, \gamma)$, corresponding to the constraints on polynomials:
\begin{equation}
    y^\alpha-1=0, \quad x^\beta y^\gamma-1=0.
\end{equation}
Accordingly, Theorem~\ref{thm: maximal logical dimension} on a twisted torus becomes:
\begin{theorem}
    On the torus with the twisted periodic boundary condition labeled by $\vec{a}_1 = (0, \alpha)$ and $\vec{a}_2 = (\beta, \gamma)$, the stabilizer codes parameterized by polynomials $f(x,y)$ and $g(x,y)$ in Eq.~\eqref{eq: stabilizer} has logical dimension
    \begin{equation}
        k = 2 \dim \left(
        \frac{\mathbb{Z}_2[x,y,x^{-1},y^{-1}]}{\langle f(x,y),~g(x,y),~y^\alpha - 1,~x^\beta y^\gamma - 1 \rangle}
        \right).
    \label{eq: k formula}
    \end{equation}
\label{thm: k on twisted torus}
\end{theorem}
{\change
Theorem~\ref{thm: k on twisted torus} extends Corollary 4.5 of Ref.~\cite{haah_module_13} to twisted tori.}
This theorem allows us to compute the code parameters without the need to construct large parity-check matrices, whose rank computation is typically costly, in contrast to the more efficient Gr\"obner basis method described above.

Next, we consider several applications of this theorem.
\begin{example}
{\bf $\mathbf{[[144, 12, 12]]}$ code.} We consider the previous Example~\ref{example: -1 -3 3 -1 generalized TC on infinite plane} of the $(3,-3)$-BB code on a (untwisted) $12 \times 6$ torus.
According to Theorem~\ref{thm: k on twisted torus}, we compute the Gr\"obner basis of the ideal
\begin{equation}
    \langle x + x^2 + y^{-3}, ~y + y^2 + x^{3},
    ~x^{12}-1, ~y^6-1 \rangle,
\end{equation}
and verify that it can be generated by the following polynomials:
\begin{eqs}
    h(x,y) &= 1 + y^2 + y^4, \\
    i(x,y) &= 1 + x + x y + x y^2 + y^3, \\
    j(x,y) &= x + x^2 + y^3.
\label{eq: G basis of 144 12 12}
\end{eqs}
From the polynomials, we can identify $6$ independent monomials:  
\begin{equation}
    1, y, y^2, y^3, x, xy,
\end{equation}
implying $k = 12$.
We can compute the code distance $d$ using either the syndrome matching algorithm~\cite{chen2022error}, the integer programming approach~\cite{landahl2011fault}, or the probabilistic algorithm~\cite{Pryadko2022GAP}.
This code becomes the $[[144, 12, 12]]$ example in Table~\ref{tab: n_k_d 1}.
\label{example: 2 3 -3 2 generalized TC on 12x6 torus}
\end{example}

We can place this $(3,-3)$-BB code on a $12 \times 12$ torus instead of the $12 \times 6$ torus, resulting in the $[[288, 12, 18]]$ code as shown in the example below. It is important to note that our $[[288, 12, 18]]$ code differs from the $[[288, 12, 18]]$ code in Ref.~\cite{Bravyi2024HighThreshold}, which uses $f(x, y) = x^3 + y^2 + y^7$ and $g(x, y) = x + x^2 + y^3$. The stabilizers in Example~\ref{example: -1 -3 3 -1 generalized TC on infinite plane} are more local, which should provide an advantage for physical implementation.
\begin{example}
{\bf $\mathbf{[[288, 12, 18]]}$ code.}
We revisit Example~\ref{example: -1 -3 3 -1 generalized TC on infinite plane} of the $(3,-3)$-BB code placed on a $12 \times 12$ torus.
We compute the Gr\"obner basis of the ideal
\begin{equation}
    \langle x + x^2 + y^{-3}, ~y + y^2 + x^{3}, ~x^{12}-1, ~y^{12} -1 \rangle,
\end{equation}
and find that it yields the same result as in Eq.~\eqref{eq: G basis of 144 12 12}. Therefore, the logical dimension is $k = 12$. Finally, we compute the code distance $d$ and confirm that this gives the $[[288, 12, 18]]$ code presented in Table~\ref{tab: n_k_d 2}.
\label{example: [[288, 12, 18]] generalized TC on twisted torus}
\end{example}

Moreover, we can put the same stabilizers on a twisted torus and obtain a different code:
\begin{example}
{\bf $\mathbf{[[270, 8, 20]]}$ code.}
We revisit Example~\ref{example: -1 -3 3 -1 generalized TC on infinite plane} of the $(3,-3)$-BB code, now defined on a twisted $9 \times 15$ torus with basis vectors $\vec{a}_1 = (0,15)$ and $\vec{a}_2 = (9,6)$
We then compute the Gr\"obner basis of the ideal
\begin{equation}
    \langle x + x^2 + y^{-3}, ~y + y^2 + x^{3}, ~x^9 y^{6} -1, ~y^{15} -1 \rangle.
\end{equation}
and find that it is generated by the following polynomials:
\begin{eqs}
    h(x,y) &= 1 + y + y^2, \\
    i(x,y) &= 1 + x + x^2.
\end{eqs}
From this, we identify the $4$ independent monomials:
\begin{equation}
    1, x, y, xy.
\end{equation}
This confirms that the code has $k = 8$ logical qubits. Additionally, we compute the code distance $d$ on the twisted torus and verify that this corresponds to the $[[270, 8, 20]]$ example listed in Table~\ref{tab: n_k_d 2}.
\label{example: [[270, 8, 20]] generalized TC on twisted torus}
\end{example}
We found that the $(3,-3)$-BB code in Example~\ref{example: -1 -3 3 -1 generalized TC on infinite plane} generates optimal codes on various twisted tori, including
\begin{eqs}
    &[[90, 8, 10]], [[108, 8, 10]], [[144, 12, 12]], [[162, 8, 14]], \\
    &[[168, 8, 14]], [[180, 8, 16]], [[234, 8, 18]], [[270, 8, 20]], \\
    &[[288, 12, 18]], [[306, 8, \leq 22]], [[324, 8, \leq 22]], [[342, 8, \leq 22]], \\
    &[[396,8,\leq 26]].
\label{eq: list of 3 -3 BB codes}
\end{eqs}
The corresponding lattice details are provided in Tables~\ref{tab: n_k_d 1}, \ref{tab: n_k_d 2}, \ref{tab: n_k_d 3}, and~\ref{tab: n_k_d 4}.

Similarly, we observe that the $(3,3)$-BB code in Example~\ref{example: -1 3 3 -1 generalized TC on infinite plane} also generates another large family of codes on finite tori, such as
\begin{eqs}
    &[[72, 8, 8]], [[108, 8, 10]], [[144, 12, 12]], [[162, 8, 14]], \\
    &[[180, 8, 16]], [[192, 8, 16]], [[234, 8, 18]], [[270, 8, 20]], \\
    &[[282, 4, \leq 24]], [[360, 12, \leq 24]].
\label{eq: list of 3 3 BB codes}
\end{eqs}

\begin{center}
\textit{Logical operators for generic $k$}
\end{center}

We now turn to the frustrated-$k$ regime ($k<k_{\mathrm{max}}$). For a complete mathematical derivation and its interpretation from algebraic geometry, see Ref.~\cite{chen2025anyon}.

As in Eq.~\eqref{eq: m-string-anyon}, we seek a string operator that transports the anyon $v_m= [0,a(x,y)]$ along the $(0,\alpha)$ cycle. However, when $(0,\alpha)$ is not compatible with the intrinsic anyon periodicity, a solution may not exist. To address this, we relax the anyon equation to
\begin{eqs}
  &(y^\alpha - 1)\,a(x,y) + (x^\beta y^\gamma - 1)\,b(x,y) \\
  &= p_1(x,y)\,g(x,y) + p_2(x,y)\,f(x,y),
\label{eq: frustrated-anyon-equation}
\end{eqs}
where the logical operator is now characterized by a pair of anyons, $[0, a(x,y)]$ and $[0, b(x,y)]$, moving along the $(0, \alpha)$ and $(\beta, \gamma)$ directions, respectively. Eq.~\eqref{eq: frustrated-anyon-equation} can be solved locally around strips along the $(0,\alpha)$ and $(\beta,\gamma)$ segments using the modified Gaussian elimination algorithm developed in Ref.~\cite{liang2023extracting}. This method has a computational complexity of $O((\alpha+\beta+\gamma)^3)$, which is faster than solving for logical operators via the parity-check matrix, whose complexity is $O(n^3) = O((\alpha\beta)^3)$.

\section{Applications for qLDPC code constructions}

Using Theorem~\ref{thm: k on twisted torus}, we systematically searched for all generalized toric codes in Eq.~\eqref{eq: a b c d generalized toric code} with twisted periodic boundary conditions for $n \leq 400$. For each even $n$, we first identified all decompositions of the form $n = 2l \times m$ and defined the twisted torus using the basis vectors $\vec{a}_1 = (0, m)$ and $\vec{a}_2 = (l, q)$, where $0 \leq q < m$. Next, we enumerated all polynomials\footnote{In principle, our search can be extended to more general polynomials of the form
\begin{eqs}
    f(x,y) &= 1 + x^{a_1} y^{b_1} + x^{a_2} y^{b_2}, \\
    g(x,y) &= 1 + x^{c_1} y^{d_1} + x^{c_2} y^{d_2}.
\end{eqs}
However, tests for small values of $n \leq 108$ indicate that these more general BB codes do not yield better $[[n, k, d]]$ parameters on twisted tori. Therefore, we restricted our search to the generalized toric codes defined in Eq.~\eqref{eq: a b c d generalized toric code} for simplicity.}
\begin{eqs}
    f(x,y) &= 1 + x + x^a y^b, \\
    g(x,y) &= 1 + y + x^c y^d,
\end{eqs}
with the exponent pairs $(a,b)$ and $(c,d)$ lying within the parallelogram spanned by $\vec{a}_1$ and $\vec{a}_2$. We then computed the corresponding $[[n, k, d]]$ parameters using the methods described in Example~\ref{example: 2 3 -3 2 generalized TC on 12x6 torus}.

Note that the computation of $k$ in Eq.~\eqref{eq: k formula} is not particularly sensitive to the system size because we can reduce $y^\alpha - 1$ and $x^\beta y^\gamma - 1$ modulo $f(x,y)$ and $g(x,y)$, retaining only the remainders.
In contrast to the conventional approach of computing $k$ via the rank of the parity-check matrix~\cite{Calderbank1996Good}---whose size scales with $n$---our method applies Gaussian elimination to a set of monomials whose range is bounded by $O(k)$, as shown in Eq.~\eqref{eq: range of monomials}.
During our search, in over $90\%$ of cases the Gr\"obner basis computation immediately yields $\langle 1 \rangle$, implying that $k=0$. Consequently, constructing parity-check matrices is unnecessary in these instances, significantly reducing the overall computational workload.
As a result, our approach requires substantially fewer computational resources for large $n$, allowing us to systematically explore codes up to $n = 400$. All computations were performed on a standard personal computer, demonstrating that the required computational resources are modest.

The optimal results for each $n$ are summarized in Tables~\ref{tab: n_k_d 1}, \ref{tab: n_k_d 2}, \ref{tab: n_k_d 3}, and~\ref{tab: n_k_d 4}. Among the various choices of $f(x,y)$, $g(x,y)$, $\vec{a}_1$, and $\vec{a}_2$ that yield the same $[[n, k, d]]$ parameters, we selected the one with the most local stabilizers. Although several $[[n, k, d]]$ codes have been reported in the literature~\cite{Bravyi2024HighThreshold, tiew2024low, wang2024coprime, eberhardt2024logical}, the codes presented in our tables exhibit improved locality on twisted tori, with the degrees of the polynomials $f(x,y)$ and $g(x,y)$ being lower than those in previous constructions.
{\change
For codes with $d \leq 20$, the code distance can be computed exactly; for codes with larger $d$, we employed a probabilistic algorithm with sufficient runtime to obtain an upper bound that we believe to be tight.}
Since each $[[n, k, d]]$ code in Tables~\ref{tab: n_k_d 1}, \ref{tab: n_k_d 2}, \ref{tab: n_k_d 3}, and~\ref{tab: n_k_d 4} usually has dozens or even hundreds of solutions, we are confident that the reported parameters accurately reflect the optimal generalized toric codes.

\subsection{Novel $[[n, k, d]]$ codes}

The novel codes are listed in Tables~\ref{tab: n_k_d 1}, \ref{tab: n_k_d 2}, \ref{tab: n_k_d 3}, and \ref{tab: n_k_d 4}, with the parameters $[[n, k, d]]$ presented in bold. For small $n$, the code
\begin{equation}
    [[120, 8, 12]]: \quad \frac{kd^2}{n} = 9.6,
\end{equation}
 appears to be absent in the existing literature (to our best knowledge).
This index follows from the Bravyi-Poulin-Terhal (BPT) bound, which states that any two-dimensional geometrically local quantum code must satisfy~\cite{bravyi2009no, Bravyi2010Tradeoffs}:
\begin{equation}
    kd^2 = O(n).
\end{equation}
For instance, the Kitaev toric code on a $L \times L$ torus scales as
\begin{equation}
[[n, k, d]] = [[2L^2, 2, L]]: \quad \frac{kd^2}{n} = 1.
\end{equation}
Thus, the value of $kd^2/n$ serves as a measure of the performance of a two-dimensional quantum code compared to the Kitaev toric code.
The code $[[120, 8, 12]]$ is currently the best known example below $n = 144$, at which the gross code $[[144, 12, 12]]$ was previously proposed~\cite{Bravyi2024HighThreshold}.

Other notable examples include the following codes:
\begin{eqs}
    [[254, 14, 16]]:& \quad \frac{kd^2}{n} = 14.11, \\
    [[294, 10, 20]]:& \quad \frac{kd^2}{n} = 13.61. \\
\end{eqs}
These codes surpass the previously best-reported weight-6 code $[[288, 12, 18]]$ for comparable system sizes, which achieves $kd^2/n = 13.5$~\cite{Bravyi2024HighThreshold}.

Another noteworthy example is the code:
\begin{equation}
    [[310, 10, 22]]:\quad \frac{kd^2}{n} = 15.61,
\end{equation}
presented in Table~\ref{tab: n_k_d 4} around $n \approx 300$ physical qubits. Its stabilizers are given by
\begin{eqs}
    f(x,y) &= 1 + x + x^3 y^2, \\
    g(x,y) &= 1 + y + x^{-4} y^4,
\end{eqs}
and are implemented on a twisted torus characterized by lattice vectors $\vec{a}_1 = (0,31)$ and $\vec{a}_2 = (5, 11)$. We observe that optimal $[[n, k, d]]$ codes for each given $n$ are frequently realized on twisted tori.

\subsection{Improved locality of stabilizers in comparison to previous constructions}

As discussed before Example~\ref{example: [[288, 12, 18]] generalized TC on twisted torus}, our construction is more localized than previous ones in the literature. In this section, we focus on another important example:
\begin{equation}
    [[360, 12, 24]]: \quad \frac{kd^2}{n} = 19.2.
\end{equation}
This code was first proposed in Ref.~\cite{Bravyi2024HighThreshold} using the polynomials
\begin{eqs}
    f(x,y) &= x^{25} + x^{26} + y^3, \\
    g(x,y) &= y + y^2 + x^9,
\end{eqs}
and implemented on an untwisted $30 \times 6$ torus. The stabilizers have a range of $9$ in the $x$-direction (since $x^{25}$ and $x^{26}$ can be equivalently treated as $x^{-5}$ and $x^{-4}$).

In contrast, our $[[360, 12, 24]]$ code is simply the $(3,3)$-BB code (Example~\ref{example: -1 3 3 -1 generalized TC on infinite plane}), specified by the following polynomials:
\begin{eqs}
    f(x, y) &= x + x^2 + y^3, \\
    g(x, y) &= y + y^2 + x^3,
\end{eqs}
placed on a twisted $6 \times 30$ torus with lattice vectors $\vec{a}_1 = (0, 30)$ and $\vec{a}_2 = (6, 6)$.
For physical realization, once we have the architecture for the stabilizers of the $(3,3)$-BB code, we can generate optimal generalized toric codes on various lattices, as listed in Eq.~\eqref{eq: list of 3 3 BB codes}.

Similarly, the physical construction of the $(3,-3)$-BB code (Example~\ref{example: -1 -3 3 -1 generalized TC on infinite plane}) can generate the quantum LDPC codes listed in Eq.~\eqref{eq: list of 3 -3 BB codes}. By comparing the stabilizers in Tables~\ref{tab: n_k_d 1}, \ref{tab: n_k_d 2}, \ref{tab: n_k_d 3}, and~\ref{tab: n_k_d 4} with those in the literature, we observe that twisted tori generally reduce the range of stabilizers, making experimental realization more feasible.

\subsection{Relation to one-dimensional generalized bicycle codes}

We present another example from Table~\ref{tab: n_k_d 3}, the $[[254, 14, 16]]$ code, which achieves $kd^2/n = 14.11$. This code is defined on a twisted $1 \times 127$ torus with lattice vectors $\vec{a}_1 = (0, 127)$ and $\vec{a}_2 = (1, 25)$. The associated polynomials are:
\begin{eqs}
    f(x,y) &= 1 + x + x^{-1}y^{-3}, \\
    g(x,y) &= 1 + y + y^{-6}.
\end{eqs}
The code is local on the twisted torus, as the range of each stabilizer is small relative to the total system size $n$. Since the twisted torus is narrow in the $x$-direction, we can remove the $x$-direction periodicity by using the polynomial $xy^{25} - 1$ to cancel the $x$-dependence. Therefore, we can reduce the code to a non-local one-dimensional quantum code. This transformation yields the following polynomials:
\begin{eqs}
    f(y) &= 1 + y^{22} + y^{102}, \\
    g(y) &= 1 + y + y^{121},
\label{eq: 254 14 16 GB code}
\end{eqs}
with a periodic boundary condition $y^{127} - 1 = 0$.
In these one-dimensional codes, the Gr\"obner basis in Theorem~\ref{thm: k on twisted torus} reduces to the $\gcd$ (greatest common divisor) for univariate polynomials, simplifying to the following expression:
\begin{eqs}
    k &= 2 \dim\left(
        \frac{\mathbb{Z}_2[y, y^{-1}]}{\langle f(y),~g(y),~y^l - 1 \rangle}
        \right) \\
    &= 2 \deg\left(\gcd \left(f(y),~g(y),~y^l - 1 \right)\right),
\label{eq: k of 1d GB codes}
\end{eqs}
which precisely matches Proposition 1 in Ref.~\cite{panteleev2021degenerate}, which computes the logical dimension of the generalized bicycle (GB) codes~\cite{Kovalev2013QuantumKronecker, Pryadko2022DistanceGB}.
For each value of $n$, we can apply the same procedure to the generalized toric codes on the twisted $1 \times \frac{n}{2}$ tori:
\begin{equation}
    \vec{a}_1 = (0, \frac{n}{2}), \quad \vec{a}_2 = (1, \gamma), \quad \text{with } 0 \leq \gamma < \frac{n}{2},
\end{equation}
to induce the corresponding one-dimensional generalized bicycle codes. The results are summarized in Table~\ref{tab: n_k_d 5}, \ref{tab: n_k_d 6}, \ref{tab: n_k_d 7}, and~\ref{tab: n_k_d 8}.

For comparison, consider the GB code described in Ref.~\cite{panteleev2021degenerate}. The polynomials for this GB code are:
\begin{eqs}
    f(y) &= 1 + y^{15} + y^{20} + y^{28} + y^{66}, \\
    f(y) &= 1 + y^{58} + y^{59} + y^{100} + y^{121},
\label{eq: 254 28 14_20 GB code}
\end{eqs}
defined on a cycle of length $l = 127$. This GB code uses weight-10 stabilizers to achieve better code parameters $[[254, 28, 14 \leq d \leq 20]]$.

\renewcommand{\arraystretch}{1.2}
\begin{table}[t]
\centering
\begin{tabular}{|c|c|c|c|}
\hline
$[[n,k,d]]$ & $f(y)$
& $g(y)$ & $l$     

\\ \hline
~$[[{\color{red}12,4,2}]]$~ & ~$1+y+y^{2}$~ & ~$1+y+y^{2}$~ &
~$6$~
\\ \hline
 
~$[[{\color{red}14,6,2}]]$~ & ~$1+y+y^{3}$~ & ~$1+y+y^{3}$~ &
~$7$~
\\ \hline
 
~$[[{\color{red}18,4,4}]]$~ & ~$1+y^{2}+y^{4}$~ & ~$1+y+y^{2}$~ &
~$9$~
\\ \hline
 
~$[[{\color{red}24,4,4}]]$~ & ~$1+y^{2}+y^{4}$~ & ~$1+y+y^{2}$~ &
~$12$~
\\ \hline
 
~$[[{\color{red}28,6,4}]]$~ & ~$1+y^{2}+y^{3}$~ & ~$1+y+y^{5}$~ &
~$14$~
\\ \hline
 
~$[[30,8,4]]$~ & ~$1+y^{2}+y^{8}$~ & ~$1+y+y^{4}$~ &
~$15$~
\\ \hline
 
~$[[{\color{red}36,4,6}]]$~ & ~$1+y^{2}+y^{4}$~ & ~$1+y+y^{5}$~ &
~$18$~
\\ \hline
 
~$[[42,10,4]]$~ & ~$1+y^{2}+y^{10}$~ & ~$1+y+y^{5}$~ &
~$21$~
\\ \hline
 
~$[[{\color{red}48,4,8}]]$~ & ~$1+y^{5}+y^{7}$~ & ~$1+y+y^{5}$~ &
~$24$~
\\ \hline
 
~$[[54,4,8]]$~ & ~$1+y^{5}+y^{7}$~ & ~$1+y+y^{5}$~ &
~$27$~
\\ \hline
 
~$[[{\color{red}56,6,8}]]$~ & ~$1+y^{3}+y^{9}$~ & ~$1+y+y^{5}$~ &
~$28$~
\\ \hline
 
~$[[{\color{red}60,8,6}]]$~ & ~$1+y^{7}+y^{9}$~ & ~$1+y+y^{4}$~ &
~$30$~
\\ \hline
 
~$[[{\color{red}62,10,6}]]$~ & ~$1+y^{3}+y^{8}$~ & ~$1+y+y^{12}$~ &
~$31$~
\\ \hline
 
~$[[{\color{red}66,4,10}]]$~ & ~$1+y^{2}+y^{7}$~ & ~$1+y+y^{11}$~ &
~$33$~
\\ \hline
 
~$[[{\color{red}70,6,8}]]$~ & ~$1+y^{3}+y^{9}$~ & ~$1+y+y^{5}$~ &
~$35$~
\\ \hline
 
~$[[72,4,10]]$~ & ~$1+y^{2}+y^{7}$~ & ~$1+y+y^{11}$~ &
~$36$~
\\ \hline
 
~$[[{\color{red}78,4,10}]]$~ & ~$1+y^{5}+y^{7}$~ & ~$1+y+y^{8}$~ &
~$39$~
\\ \hline
 
~$[[84,10,6]]$~ & ~$1+y^{11}+y^{13}$~ & ~$1+y+y^{5}$~ &
~$42$~
\\ \hline
 
~$[[90,8,8]]$~ & ~$1+y^{2}+y^{9}$~ & ~$1+y+y^{12}$~ &
~$45$~
\\ \hline
 
~$[[{\color{red}96,4,12}]]$~ & ~$1+y^{5}+y^{7}$~ & ~$1+y+y^{11}$~ &
~$48$~
\\ \hline
 
~$[[{\color{red}98,6,12}]]$~ & ~$1+y^{4}+y^{12}$~ & ~$1+y+y^{10}$~ &
~$49$~
\\ \hline
 
~$[[{\color{red}102,4,12}]]$~ & ~$1+y^{4}+y^{8}$~ & ~$1+y+y^{11}$~ &
~$51$~
\\ \hline
 
~$[[108,4,12]]$~ & ~$1+y^{8}+y^{10}$~ & ~$1+y+y^{8}$~ &
~$54$~
\\ \hline

\end{tabular}
\caption{One-dimensional generalized bicycle codes induced from generalized toric codes on twisted $1 \times \frac{n}{2}$ tori for $ n \leq 110$.
The code is defined on a circle of length $l = \frac{n}{2}$, with two qubits per unit cell. The red notation $[[{\color{red} n, k, d}]]$ indicates that the code parameters are identical to those of the optimal generalized toric code in two dimensions.}
\label{tab: n_k_d 5}
\end{table}

\renewcommand{\arraystretch}{1.2}
\begin{table}[t]
\centering
\begin{tabular}{|c|c|c|c|}
\hline
$[[n,k,d]]$ & $f(y)$
& $g(y)$ & $l$     

\\ \hline

~$[[{\color{red}112,6,12}]]$~ & ~$1+y^{3}+y^{15}$~ & ~$1+y+y^{10}$~ &
~$56$~
\\ \hline
 
~$[[{\color{red}114,4,14}]]$~ & ~$1+y^{8}+y^{13}$~ & ~$1+y+y^{11}$~ &
~$57$~
\\ \hline
 
~$[[{\color{red}120,8,12}]]$~ & ~$1+y^{8}+y^{21}$~ & ~$1+y+y^{12}$~ &
~$60$~
\\ \hline
 
~$[[{\color{red}124,10,10}]]$~ & ~$1+y^{8}+y^{11}$~ & ~$1+y+y^{13}$~ &
~$62$~
\\ \hline
 
~$[[{\color{red}126,12,10}]]$~ & ~$1+y^{12}+y^{23}$~ & ~$1+y+y^{8}$~ &
~$63$~
\\ \hline
 
~$[[{\color{red}132,4,14}]]$~ & ~$1+y^{4}+y^{14}$~ & ~$1+y+y^{14}$~ &
~$66$~
\\ \hline
 
~$[[{\color{red}138,4,14}]]$~ & ~$1+y^{8}+y^{13}$~ & ~$1+y+y^{8}$~ &
~$69$~
\\ \hline
 
~$[[{\color{red}140,6,14}]]$~ & ~$1+y^{10}+y^{16}$~ & ~$1+y+y^{12}$~ &
~$70$~
\\ \hline

~$[[144,4,16]]$~ & ~$1+y^{23}+y^{28}$~ & ~$1+y+y^{20}$~ &
~$72$~
\\ \hline
 
~$[[{\color{red}146,18,4}]]$~ & ~$1+y^{2}+y^{18}$~ & ~$1+y+y^{9}$~ &
~$73$~
\\ \hline
 
~$[[{\color{red}150,8,12}]]$~ & ~$1+y^{2}+y^{8}$~ & ~$1+y+y^{19}$~ &
~$75$~
\\ \hline
 
~$[[{\color{red}154,6,16}]]$~ & ~$1+y^{4}+y^{34}$~ & ~$1+y+y^{19}$~ &
~$77$~
\\ \hline
 
~$[[{\color{red}156,4,16}]]$~ & ~$1+y^{11}+y^{16}$~ & ~$1+y+y^{14}$~ &
~$78$~
\\ \hline
 
~$[[162,4,16]]$~ & ~$1+y^{7}+y^{11}$~ & ~$1+y+y^{14}$~ &
~$81$~
\\ \hline
 
~$[[168,10,12]]$~ & ~$1+y^{11}+y^{19}$~ & ~$1+y+y^{17}$~ &
~$84$~
\\ \hline
 
~$[[{\color{red}170,16,10}]]$~ & ~$1+y^{21}+y^{25}$~ & ~$1+y+y^{16}$~ &
~$85$~
\\ \hline

~$[[{\color{red}174,4,18}]]$~ & ~$1+y^{7}+y^{11}$~ & ~$1+y+y^{17}$~ &
~$87$~
\\ \hline
 
~$[[{\color{red}180,8,16}]]$~ & ~$1+y^{8}+y^{47}$~ & ~$1+y+y^{34}$~ &
~$90$~
\\ \hline
 
~$[[{\color{red}182,6,18}]]$~ & ~$1+y^{9}+y^{13}$~ & ~$1+y+y^{38}$~ &
~$91$~
\\ \hline
 
~$[[186,14,10]]$~ & ~$1+y^{8}+y^{19}$~ & ~$1+y+y^{14}$~ &
~$93$~
\\ \hline
 
~$[[192,4,18]]$~ & ~$1+y^{11}+y^{16}$~ & ~$1+y+y^{14}$~ &
~$96$~
\\ \hline
 
~$[[{\color{red}196,6,18}]]$~ & ~$1+y^{12}+y^{22}$~ & ~$1+y+y^{19}$~ &
~$98$~
\\ \hline

\end{tabular}
\caption{Continuation of Table~\ref{tab: n_k_d 5} for $110 < n \leq 196$.}
\label{tab: n_k_d 6}
\end{table}

\section{Discussion and future directions}

We have introduced a topological order perspective to studying quantum error-correcting codes on tori. 
From the algebraic structure of anyons, the logical dimension $k$ can be determined by counting independent anyon types. We showed that this corresponds to the dimension of the quotient ring $R/I$, where the ideal $I = \langle f(x,y),~g(x,y) \rangle$ is generated by the stabilizers.
This provides a systematic approach to characterizing the code space.
Our framework naturally incorporates (twisted) periodic boundary conditions, enabling the construction and characterization of new quantum LDPC codes. To ensure computational feasibility, we employed Gr\"obner basis techniques, enabling a systematic analysis of generalized toric codes up to $n \leq 400$ physical qubits.
The versatility of our method is reflected in the discovery of novel qLDPC codes listed in Tables~\ref{tab: n_k_d 1}, \ref{tab: n_k_d 2}, \ref{tab: n_k_d 3}, and~\ref{tab: n_k_d 4}. These results illustrate the power of a ring-theoretic approach in advancing the understanding of topological quantum codes, paving the way for future explorations in both theory and practical implementation.

Future work could extend this investigation to larger system sizes (higher $n$), as these may yield improved codes.
Given that our search algorithm is fully parallelizable, supercomputers or computer clusters could be employed to examine all generalized toric codes within $n \leq 500$ or higher---scales that are comparable to the number of physical qubits in state-of-the-art experimental platforms~\cite{Google2019supremacy, Pan2020Quantum, Lukin2021programmable, Madsen2022Quantum, Bravyi2022future, Lukin2025Quantum}. The primary bottleneck, however, is the computation of code distances. When $n$ reaches a few hundred and $d$ exceeds 20, the probabilistic algorithm for computing the code distance may not be reliable and could only yield an upper bound for $d$.

Alternatively, one could explore different forms of polynomials. For example, Ref.~\cite{wang2024coprime} considered
\begin{eqs}
    f(x,y) &= 1 + (xy)^{a'} + (xy)^{b'}, \\
    g(x,y) &= 1 + (xy)^{c'} + (xy)^{d'},
\end{eqs}
treating $\pi := xy$ as a single variable. However, since the exponents $a'$, $b'$, $c'$, and $d'$ range from $0$ to $n$---a considerably larger interval than that in Eq.~\eqref{eq: a b c d generalized toric code}---the required computational resources are substantially higher. It would be interesting to compare the resulting code parameters with those presented in this work. Additionally, one could increase the weights of stabilizers, which could generate improved $[[n, k, d]]$ parameters, as demonstrated in Eq.~\eqref{eq: 254 28 14_20 GB code} and previous constructions of quantum LDPC codes~\cite{Kovalev2013QuantumKronecker, breuckmann2021balanced, panteleev2021degenerate, Lin2022c3Locally, Leverrier2022Tanner, Panteleev2022goodqldpc, Dinur2023Good, wang2023abelian, Lin2024Quantumtwoblock, tiew2024low, Wills2024Localtestability, eberhardt2024logical, Wills2025Tradeoff}.

Finally, it is important to investigate whether the codes presented in Tables~\ref{tab: n_k_d 1}, \ref{tab: n_k_d 2}, \ref{tab: n_k_d 3}, and~\ref{tab: n_k_d 4} can achieve comparable error suppression in the context of the circuit-based noise model, as discussed in Ref.~\cite{Bravyi2024HighThreshold}. In particular, optimizing the depth of the syndrome measurement circuit is essential since the effective circuit‐level distance is typically smaller than the nominal code distance. We plan to carry out numerical simulations to estimate the pseudo‐thresholds of these codes.
Furthermore, physically realizing these codes, such as via bilayer superconducting-qubit architectures, and developing efficient logical-gate implementations are critical steps toward their deployment as practical quantum LDPC codes.

\renewcommand{\arraystretch}{1.2}
\begin{table}[t]
\centering
\begin{tabular}{|c|c|c|c|}
\hline
$[[n,k,d]]$ & $f(y)$
& $g(y)$ & $l$     

\\ \hline

~$[[198,4,18]]$~ & ~$1+y^{11}+y^{16}$~ & ~$1+y+y^{14}$~ &
~$99$~
\\ \hline
 
~$[[{\color{red}204,4,20}]]$~ & ~$1+y^{16}+y^{35}$~ & ~$1+y+y^{11}$~ &
~$102$~
\\ \hline
 
~$[[210,14,12]]$~ & ~$1+y^{11}+y^{27}$~ & ~$1+y+y^{19}$~ &
~$105$~
\\ \hline
 
~$[[216,4,20]]$~ & ~$1+y^{14}+y^{22}$~ & ~$1+y+y^{20}$~ &
~$108$~
\\ \hline
 
~$[[{\color{red}222,4,20}]]$~ & ~$1+y^{10}+y^{14}$~ & ~$1+y+y^{20}$~ &
~$111$~
\\ \hline
 
~$[[{\color{red}224,6,20}]]$~ & ~$1+y^{3}+y^{22}$~ & ~$1+y+y^{31}$~ &
~$112$~
\\ \hline
 
~$[[{\color{red}228,4,20}]]$~ & ~$1+y^{7}+y^{17}$~ & ~$1+y+y^{20}$~ &
~$114$~
\\ \hline
 
~$[[234,4,\leq 22]]$~ & ~$1+y^{13}+y^{29}$~ & ~$1+y+y^{20}$~ &
~$117$~
\\ \hline
 
~$[[{\color{red}238,6,20}]]$~ & ~$1+y^{9}+y^{20}$~ & ~$1+y+y^{24}$~ &
~$119$~
\\ \hline
 
~$[[{\color{red}240,8,18}]]$~ & ~$1+y^{13}+y^{21}$~ & ~$1+y+y^{19}$~ &
~$120$~
\\ \hline
 
~$[[{\color{red}246,4,\leq 22}]]$~ & ~$1+y^{13}+y^{20}$~ & ~$1+y+y^{23}$~ &
~$123$~
\\ \hline
 
~$[[{\color{red}248,10,18}]]$~ & ~$1+y^{17}+y^{27}$~ & ~$1+y+y^{13}$~ &
~$124$~
\\ \hline
 
~$[[{\color{red}252,12,16}]]$~ & ~$1+y^{25}+y^{30}$~ & ~$1+y+y^{8}$~ &
~$126$~
\\ \hline
 
~$[[{\color{red}254,14,16}]]$~ & ~$1+y^{10}+y^{37}$~ & ~$1+y+y^{31}$~ &
~$127$~
\\ \hline
 
~$[[{\color{red}258,4,\leq 22}]]$~ & ~$1+y^{14}+y^{19}$~ & ~$1+y+y^{14}$~ &
~$129$~
\\ \hline
 
~$[[264,4,\leq 22]]$~ & ~$1+y^{13}+y^{20}$~ & ~$1+y+y^{17}$~ &
~$132$~
\\ \hline

~$[[{\color{red}266,6,\leq 22}]]$~ & ~$1+y^{12}+y^{25}$~ & ~$1+y+y^{17}$~ &
~$133$~
\\ \hline

~$[[{\color{red}270,8,20}]]$~ & ~$1+y^{6}+y^{23}$~ & ~$1+y+y^{27}$~ &
~$135$~
\\ \hline
 
~$[[{\color{red}276,4,\leq 24}]]$~ & ~$1+y^{8}+y^{31}$~ & ~$1+y+y^{20}$~ &
~$138$~
\\ \hline
 
~$[[{\color{red}280,6,\leq 22}]]$~ & ~$1+y^{20}+y^{23}$~ & ~$1+y+y^{17}$~ &
~$140$~
\\ \hline
 
~$[[{\color{red}282,4,\leq 24}]]$~ & ~$1+y^{10}+y^{17}$~ & ~$1+y+y^{23}$~ &
~$141$~
\\ \hline
 
~$[[288,4,\leq 24]]$~ & ~$1+y^{20}+y^{25}$~ & ~$1+y+y^{14}$~ &
~$144$~
\\ \hline
 
~$[[{\color{red}292,18,8}]]$~ & ~$1+y^{4}+y^{36}$~ & ~$1+y+y^{9}$~ &
~$146$~
\\ \hline
 
\end{tabular}
\caption{Continuation of Table~\ref{tab: n_k_d 6} for $196 < n \leq 292$.}
\label{tab: n_k_d 7}
\end{table}

\renewcommand{\arraystretch}{1.2}
\begin{table}[t]
\centering
\begin{tabular}{|c|c|c|c|}
\hline
$[[n,k,d]]$ & $f(y)$
& $g(y)$ & $l$     

\\ \hline

~$[[{\color{red}294,10,20}]]$~ & ~$1+y^{19}+y^{29}$~ & ~$1+y+y^{26}$~ &
~$147$~
\\ \hline
 
~$[[{\color{red}300,8,\leq 22}]]$~ & ~$1+y^{43}+y^{52}$~ & ~$1+y+y^{57}$~ &
~$150$~
\\ \hline
 
~$[[306,4,\leq 24]]$~ & ~$1+y^{8}+y^{22}$~ & ~$1+y+y^{20}$~ &
~$153$~
\\ \hline
 
~$[[{\color{red}308,6,\leq 24}]]$~ & ~$1+y^{12}+y^{22}$~ & ~$1+y+y^{26}$~ &
~$154$~
\\ \hline
 
~$[[{\color{red}310,10,\leq 22}]]$~ & ~$1+y^{20}+y^{43}$~ & ~$1+y+y^{14}$~ &
~$155$~
\\ \hline
 
~$[[312,4,\leq 24]]$~ & ~$1+y^{10}+y^{17}$~ & ~$1+y+y^{23}$~ &
~$156$~
\\ \hline
 
~$[[{\color{red}318,4,\leq 26}]]$~ & ~$1+y^{14}+y^{34}$~ & ~$1+y+y^{38}$~ &
~$159$~
\\ \hline
 
~$[[{\color{red}322,6,\leq 24}]]$~ & ~$1+y^{5}+y^{25}$~ & ~$1+y+y^{24}$~ &
~$161$~
\\ \hline
 
~$[[324,4,\leq 26]]$~ & ~$1+y^{11}+y^{16}$~ & ~$1+y+y^{26}$~ &
~$162$~
\\ \hline

~$[[{\color{red}330,8,\leq 24}]]$~ & ~$1+y^{32}+y^{38}$~ & ~$1+y+y^{49}$~ &
~$165$~
\\ \hline
 
~$[[{\color{red}336,10,\leq 22}]]$~ & ~$1+y^{19}+y^{50}$~ & ~$1+y+y^{5}$~ &
~$168$~
\\ \hline
 
~$[[{\color{red}340,16,18}]]$~ & ~$1+y^{4}+y^{25}$~ & ~$1+y+y^{70}$~ &
~$170$~
\\ \hline
 
~$[[342,4,\leq 26]]$~ & ~$1+y^{16}+y^{23}$~ & ~$1+y+y^{20}$~ &
~$171$~
\\ \hline
 
~$[[{\color{red}348,4,\leq 26}]]$~ & ~$1+y^{16}+y^{23}$~ & ~$1+y+y^{20}$~ &
~$174$~
\\ \hline
 
~$[[{\color{red}350,6,\leq 26}]]$~ & ~$1+y^{4}+y^{33}$~ & ~$1+y+y^{24}$~ &
~$175$~
\\ \hline
 
~$[[{\color{red}354,4,\leq 28}]]$~ & ~$1+y^{19}+y^{29}$~ & ~$1+y+y^{23}$~ &
~$177$~
\\ \hline
 
~$[[360,8,\leq 24]]$~ & ~$1+y^{5}+y^{25}$~ & ~$1+y+y^{27}$~ &
~$180$~
\\ \hline
 
~$[[{\color{red}364,6,\leq 26}]]$~ & ~$1+y^{17}+y^{22}$~ & ~$1+y+y^{24}$~ &
~$182$~
\\ \hline
 
~$[[{\color{red}366,4,\leq 28}]]$~ & ~$1+y^{8}+y^{28}$~ & ~$1+y+y^{26}$~ &
~$183$~
\\ \hline
 
~$[[372,14,20]]$~ & ~$1+y^{26}+y^{34}$~ & ~$1+y+y^{20}$~ &
~$186$~
\\ \hline
 
~$[[{\color{red}378,12,\leq 22}]]$~ & ~$1+y^{4}+y^{37}$~ & ~$1+y+y^{25}$~ &
~$189$~
\\ \hline
 
~$[[384,4,\leq 28]]$~ & ~$1+y^{16}+y^{23}$~ & ~$1+y+y^{26}$~ &
~$192$~
\\ \hline
 
~$[[{\color{red}390,8,\leq 26}]]$~ & ~$1+y^{13}+y^{37}$~ & ~$1+y+y^{42}$~ &
~$195$~
\\ \hline
 
~$[[{\color{red}392,6,\leq 28}]]$~ & ~$1+y^{6}+y^{37}$~ & ~$1+y+y^{24}$~ &
~$196$~
\\ \hline
 
~$[[396,4,\leq 30]]$~ & ~$1+y^{14}+y^{22}$~ & ~$1+y+y^{32}$~ &
~$198$~
\\ \hline
\end{tabular}
\caption{Continuation of Table~\ref{tab: n_k_d 7} for $292 < n \leq 400$.}
\label{tab: n_k_d 8}
\end{table}

%%%%%%%%%%%%%%%%%%%%%%%%%%%%%%%%%%%%%%%%%%%%%%%%%%%%%%%%%%%%%%%%%%%%%%%%%%%%%%%%%%%%
\section*{Acknowledgement}

We would like to thank Arpit Dua, Jens Niklas Eberhardt, Jeongwan Haah, Zibo Jin, Zi-Wen Liu, Frank Mueller, Francisco Revson F. Pereira, Ming Wang, and Bowen Yang for their valuable discussions, and Shin Ho Choe and Vincent Steffan for confirming code distances.
This work is supported by the National Natural Science Foundation of China (Grant No.~12474491, No.~12474145, and No.~12447101).
\bigskip

%%%%%%%%%%%%%%%%%%%%%%%%%%%%%%%%%%%%%%%%%%%%%%%%%%%%%%%%%%%%%%%%%%%%%%%%%%%%%%%%%%%%
%\pagebreak
\appendix

\begin{figure*}[t]
    \centering
    \subfigure[\normalsize Boundary identification]{\includegraphics[scale=0.12]{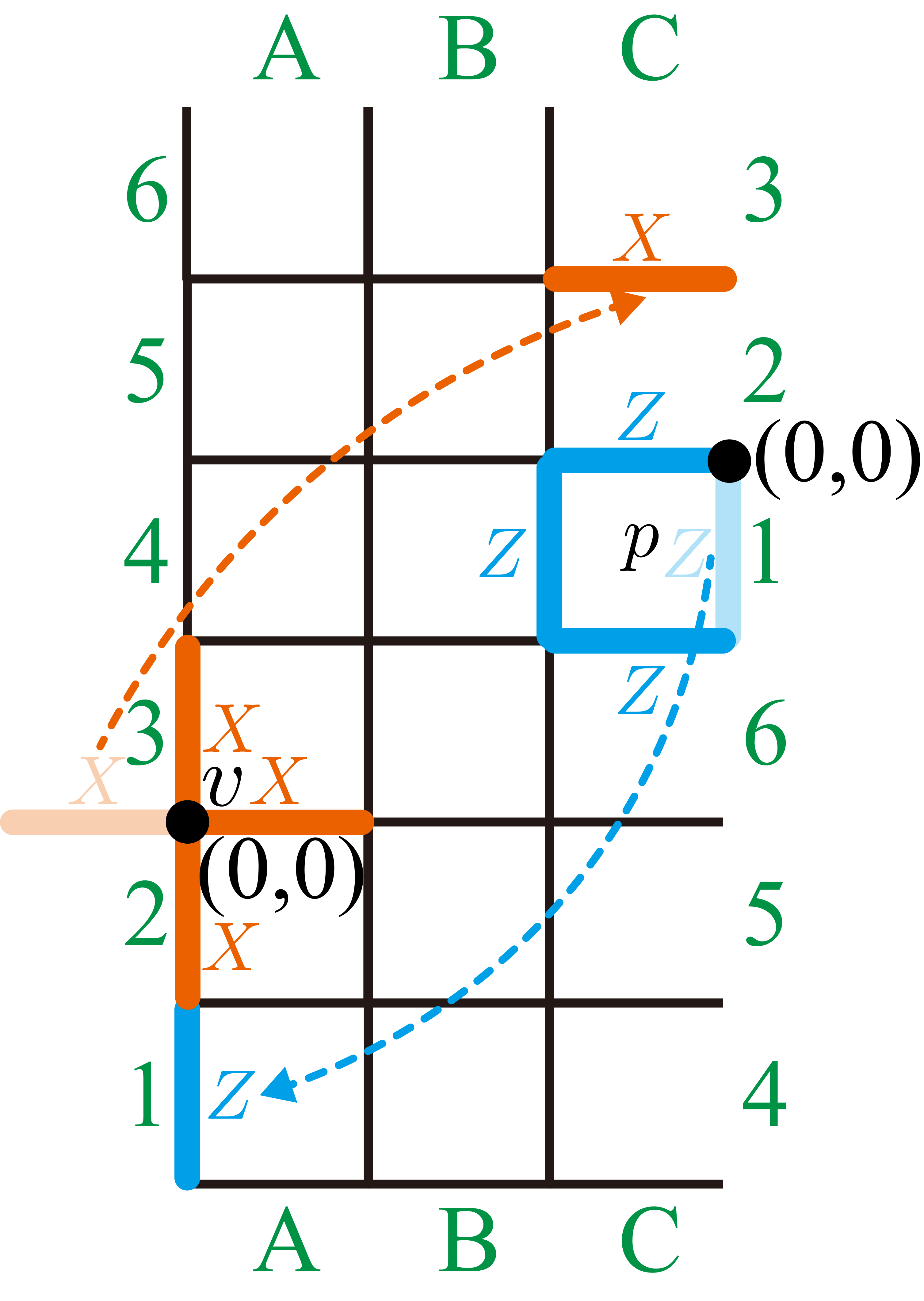}\label{fig: number_convention a}}
    \hspace{1em}
    \subfigure[\normalsize Qubit numbering]{\raisebox{-1.2ex}{\includegraphics[scale=0.12]{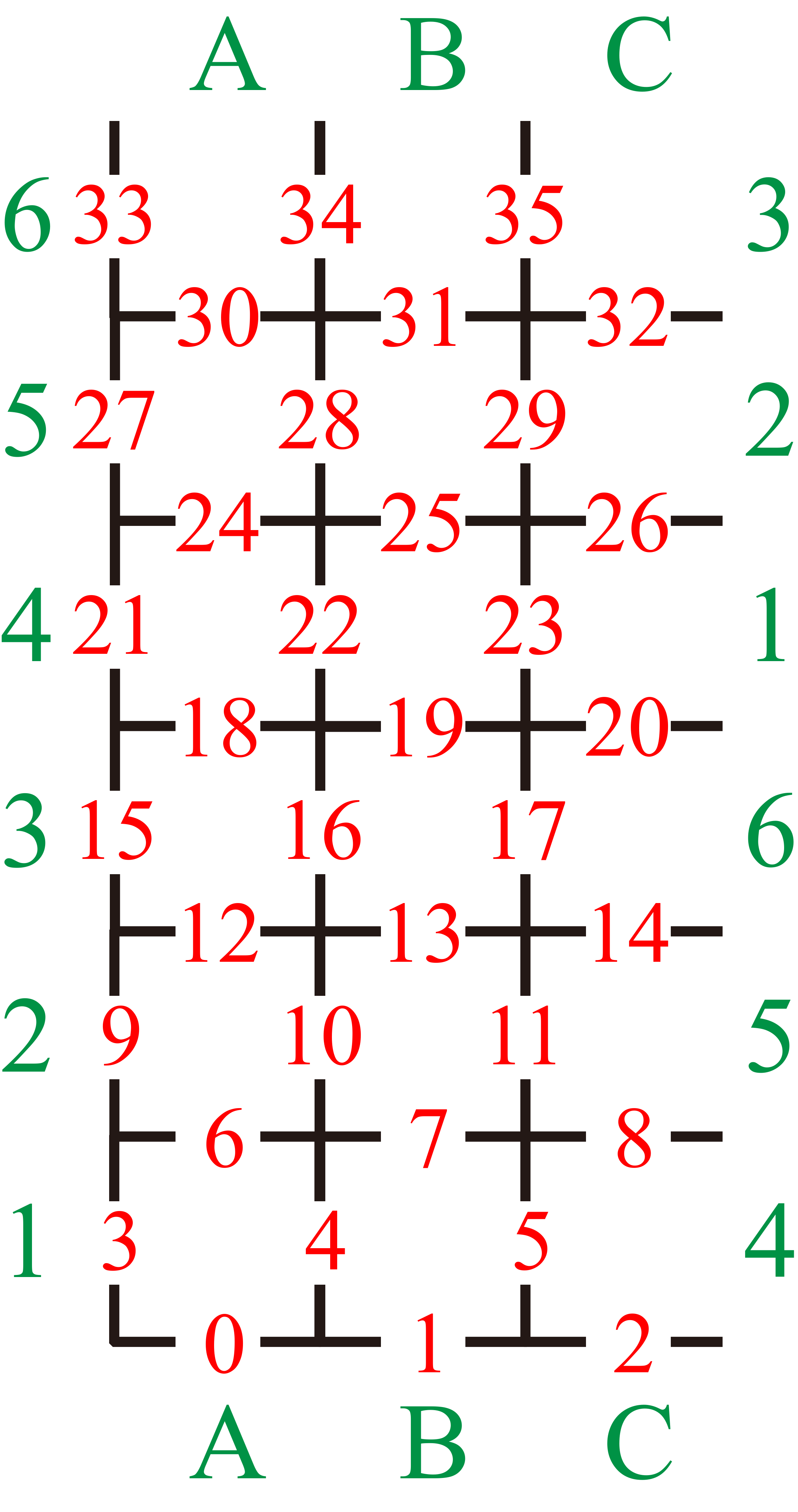}\label{fig: number_convention b}}}
    \hspace{3em}
    \subfigure[\normalsize Rotated Kitaev toric code]{\raisebox{10ex}{\includegraphics[scale=0.11]{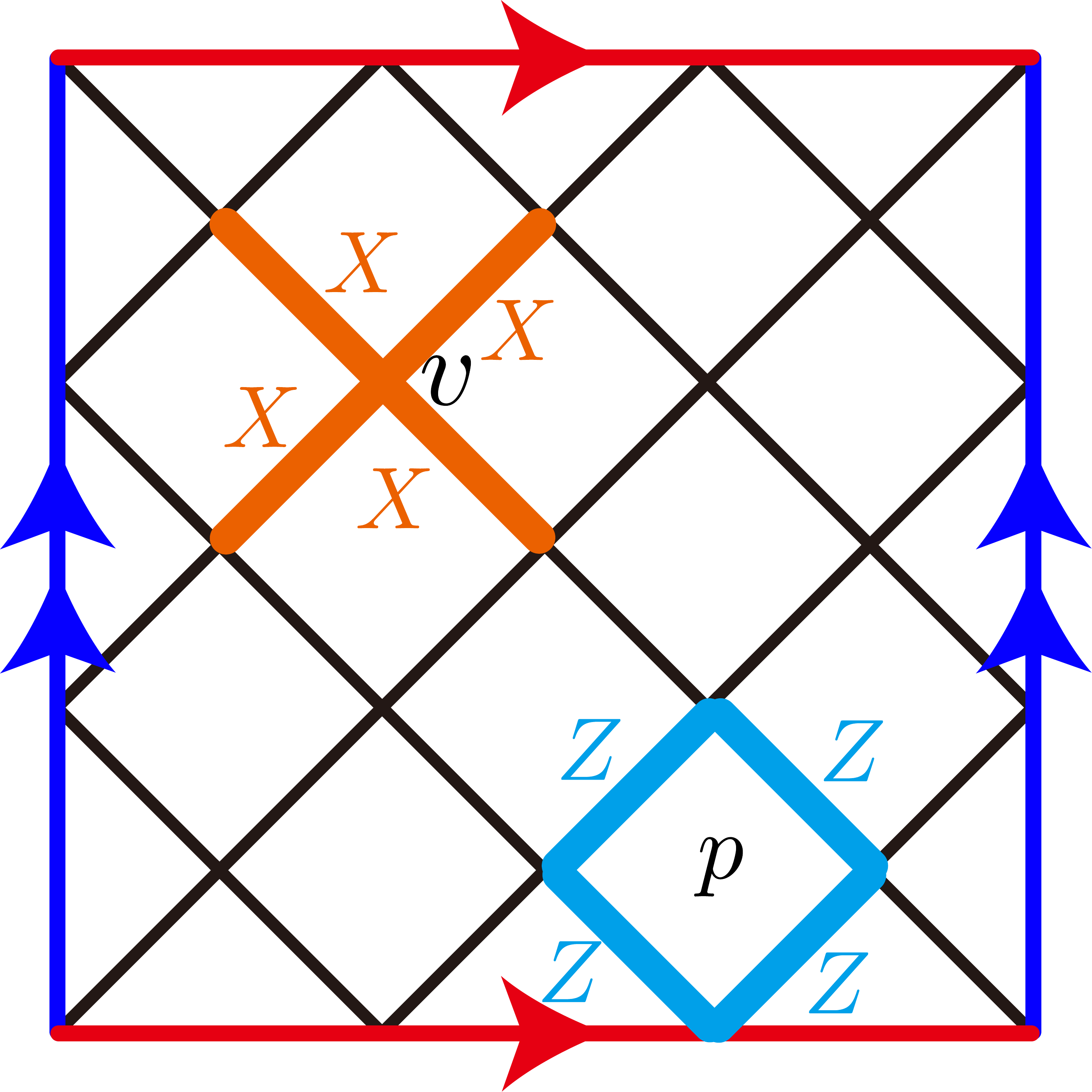}\label{fig: number_convention c}}}
    \caption{\change
    (a) Alternative representation of the twisted torus from Fig.~\ref{fig: twisted torus}, with lattice vectors $\vec{a}_1=(0,\alpha)$ and $\vec{a}_2=(\beta,\gamma)$ (here $\alpha=6$, $\beta=\gamma=3$). Green labels on the left and right edges (and top and bottom) indicate identified boundaries: for example, the horizontal edge in the red $A_v$ terms between labels 2 and 3 on the left reappears at the same location between 2 and 3 on the right, and similarly for the vertical edge in the blue $B_p$ terms at position 1.
    (b) Sequential numbering of qubits from $0$ to $2\alpha\beta-1$.
    (c) This twisted torus is equivalent to the rotated toric code, illustrating that a lattice rotation is a special case of the twisted‐torus construction.
    }
    \label{fig: number_convention}
\end{figure*}

%%%%%%%%%%%%%%%%%%%%%%%%%%%%%%%%%%%%%%%%%%%%%%%%%%%%%%%%%%%%%%%%%%%%%%%%%%%%%%%%%%%%
%\newpage

\section{One-dimensional generalized bicycle codes}
\label{sec: 1d GB codes}

As shown in Eq.~\eqref{eq: 254 14 16 GB code}, generalized toric codes on twisted $1 \times \frac{n}{2}$ tori can be reduced to one-dimensional quantum codes, specifically the generalized bicycle code~\cite{Kovalev2013QuantumKronecker, panteleev2021degenerate}. In addition to the optimal generalized toric codes in two dimensions, presented in Tables~\ref{tab: n_k_d 1}, \ref{tab: n_k_d 2}, \ref{tab: n_k_d 3}, and~\ref{tab: n_k_d 4}, we also identify the generalized bicycle (GB) codes in one dimension that are induced from these generalized toric codes. These results are summarized in Tables~\ref{tab: n_k_d 5}, \ref{tab: n_k_d 6}, \ref{tab: n_k_d 7}, and~\ref{tab: n_k_d 8}.

{\change
For the one-dimensional GB codes, the logical dimension can also be obtained from Eq.~\eqref{eq: k of 1d GB codes}, which arises as a special case of Theorem~\ref{thm: k on twisted torus}. Although twisted $1\times\frac{n}{2}$ tori naively appear as long, narrow strips, their two-dimensional realizations often admit unit cells whose aspect ratios are close to one. For example, Eq.~\eqref{eq: 1 17 0 29 from 2 5 7 3} demonstrates that the twisted torus with
$\vec{a}_1=(0,\,29), ~\vec{a}_2=(1,\,17)$
is equivalent to the one with
$\vec{a}_1=(2,\,5), ~\vec{a}_2=(7,\,3)$.

}

%%%%%%%%%%%%%%%%%%%%%%%%%%%%%%%%%%%%%%%%%%%%%%%%%%%%%%%%%%%%%%%%%%%%%%%%%%%%%%%%%%%%
\section{Implementation of Pauli operators on twisted tori}
\label{sec: Lattice implementation of twisted tori}

In this appendix, we describe our convention for representing Pauli operators on a twisted torus as binary column vectors. This convention allows us to convert the stabilizers labeled by $f(x,y)$ and $g(x,y)$—defined on a twisted torus with lattice vectors $\vec{a}_1 = (0,\alpha)$ and $\vec{a}_2 = (\beta,\gamma)$—into the $n\times n$ parity-check matrices $H_X$ and $H_Z$, corresponding to the $X$- and $Z$-type stabilizer checks, respectively. Informally, these matrices exhibit a (generalized) ``cyclic'' structure, which we now detail.

To make this structure explicit, we adopt an alternative representation of the twisted torus, shown in Fig.~\ref{fig: number_convention}.\footnote{This lattice representation is equivalent to Fig.~\ref{fig: twisted torus}.} Vertically, the system is translationally invariant along $(0,\alpha)$, so the column labels $A$, $B$, and $C$ at the top match those at the bottom. Horizontally, the lattice vector is $\vec{a}_2=(\beta,\gamma)$, and as shown in Fig.~\ref{fig: number_convention a}, the row labels on the left boundary shift by $+\gamma$ in the $y$-direction when gluing to the right boundary.

This representation facilitates systematic qubit labeling. We number the horizontal edges in the first row, then the vertical edges, then proceed to the second row, again listing the horizontal then vertical edges, and so on. Each qubit is assigned a unique index $i\in\{0,\dots,n-1\}$, where $n=2\alpha\beta$. With this labeling, shifts in the $x$- and $y$-directions are defined as follows:

\begin{enumerate}
    \item \textbf{Shifting in the $y$-direction:} From Fig.~\ref{fig: number_convention b}, moving a qubit $i$ up by one unit always adds $2\beta$. If $i+2\beta\ge n$, we subtract $n$ to keep the label in $[0,n)$. In general, a shift by $l_y$ units yields
    \begin{equation}
        i \xlongrightarrow[y]{l_y} i + (2\beta)\,l_y \pmod{n}.
    \label{eq: shift ly in y}
    \end{equation}
    \item \textbf{Shifting in the $x$-direction:} 
    If $i\not\equiv -1\pmod{\beta}$, moving one unit to the right simply increments the index by 1:
    \begin{equation}
        i \xlongrightarrow[x]{1} i+1.
    \label{eq: shifting x 1}
    \end{equation}
    However, if $i\equiv -1\pmod{\beta}$, the edge exits the lattice, so we apply the lattice vector $(-\beta,-\gamma)$ to return inside:
    \begin{equation}
        i \xlongrightarrow[x]{1} i + 1 - (2\gamma+1)\beta \pmod{n}.
    \label{eq: shifting x 2}
    \end{equation}
    Combining \eqref{eq: shifting x 1} and \eqref{eq: shifting x 2}, a shift by $l_x$ units in the $x$-direction is
    \begin{eqs}
        i \xlongrightarrow[x]{l_x} i + l_x - (2\gamma+1)\beta\bigl(\lfloor\tfrac{i+l_x}{\beta}\rfloor - \lfloor\tfrac{i}{\beta}\rfloor\bigr)\pmod{n},
    \label{eq: shift lx in x}
    \end{eqs}
    where $\lfloor r\rfloor$ denotes the greatest integer less than or equal to $r$.
\end{enumerate}

For example, in a twisted torus with $\vec{a}_1=(0,6)$ and $\vec{a}_2=(3,3)$ shown in Fig.~\ref{fig: number_convention b}, the qubit labeled $8$ shifts two units in the $y$-direction as
\begin{equation}
    8 \xlongrightarrow[y]{2} 8 + 12 = 20 \pmod{36}.
\end{equation}
Shifting qubit $8$ by two units in the $x$-direction gives
\begin{eqs}
    8 \xlongrightarrow[x]{2}& ~8+2-21\Bigl(\lfloor\tfrac{10}{3}\rfloor - \lfloor\tfrac{8}{3}\rfloor\Bigr)\pmod{36} \\
    &= -11 = 25 \pmod{36}.
\end{eqs}

\begin{widetext}

Using these shifts, we can construct the parity-check matrix efficiently. We begin by expressing a stabilizer at the origin as an $n$-dimensional column vector: for an $X$-type stabilizer acting on qubits $i_1, i_2, i_3, \dots$, the vector has ones at positions $i_1, i_2, i_3, \dots$ and zeros elsewhere. This vector is then translated by $(l_x, l_y)$ for $0 \leq l_x < \beta$ and $0 \leq l_y < \alpha$, generating $\alpha\beta = n/2$ columns that form the $n \times n/2$ matrix $H_X$. In the case of the Kitaev toric code (Fig.~\ref{fig: number_convention a}), each $A_v$ stabilizer involves four Pauli operators, so each column contains exactly four ones.
The resulting $36 \times 18$ matrix $H_X$ is given as:
\begin{eqs}
    \renewcommand{\arraystretch}{1.0}
    \begin{array}{C|*{18}{C}}
       & 1 & 2 & 3 & 4 & 5 & 6 & 7 & 8 & 9 & 10 & 11 & 12 & 13 & 14 & 15 & 16 & 17 & 18\\\hline
     0 & \textcolor{blue}{1} & \textcolor{blue}{1} & 0 & 0 & 0 & 0 & 0 & 0 & 0 & 0 & 0 & 0 & 0 & 0 & 0 & 0 & 0 & 0\\
     1 & 0 & \textcolor{blue}{1} & \textcolor{blue}{1} & 0 & 0 & 0 & 0 & 0 & 0 & 0 & 0 & 0 & 0 & 0 & 0 & 0 & 0 & 0\\
     2 & 0 & 0 & \textcolor{blue}{1} & 0 & 0 & 0 & 0 & 0 & 0 & \textcolor{blue}{1} & 0 & 0 & 0 & 0 & 0 & 0 & 0 & 0\\
     3 & \textcolor{blue}{1} & 0 & 0 & \textcolor{blue}{1} & 0 & 0 & 0 & 0 & 0 & 0 & 0 & 0 & 0 & 0 & 0 & 0 & 0 & 0\\
     4 & 0 & \textcolor{blue}{1} & 0 & 0 & \textcolor{blue}{1} & 0 & 0 & 0 & 0 & 0 & 0 & 0 & 0 & 0 & 0 & 0 & 0 & 0\\
     5 & 0 & 0 & \textcolor{blue}{1} & 0 & 0 & \textcolor{blue}{1} & 0 & 0 & 0 & 0 & 0 & 0 & 0 & 0 & 0 & 0 & 0 & 0\\
     6 & 0 & 0 & 0 & \textcolor{blue}{1} & \textcolor{blue}{1} & 0 & 0 & 0 & 0 & 0 & 0 & 0 & 0 & 0 & 0 & 0 & 0 & 0\\
     7 & 0 & 0 & 0 & 0 & \textcolor{blue}{1} & \textcolor{blue}{1} & 0 & 0 & 0 & 0 & 0 & 0 & 0 & 0 & 0 & 0 & 0 & 0\\
     8 & 0 & 0 & 0 & 0 & 0 & \textcolor{blue}{1} & 0 & 0 & 0 & 0 & 0 & 0 & \textcolor{blue}{1} & 0 & 0 & 0 & 0 & 0\\
     9 & 0 & 0 & 0 & \textcolor{blue}{1} & 0 & 0 & \textcolor{blue}{1} & 0 & 0 & 0 & 0 & 0 & 0 & 0 & 0 & 0 & 0 & 0\\
    10 & 0 & 0 & 0 & 0 & \textcolor{blue}{1} & 0 & 0 & \textcolor{blue}{1} & 0 & 0 & 0 & 0 & 0 & 0 & 0 & 0 & 0 & 0\\
    11 & 0 & 0 & 0 & 0 & 0 & \textcolor{blue}{1} & 0 & 0 & \textcolor{blue}{1} & 0 & 0 & 0 & 0 & 0 & 0 & 0 & 0 & 0\\
    12 & 0 & 0 & 0 & 0 & 0 & 0 & \textcolor{blue}{1} & \textcolor{blue}{1} & 0 & 0 & 0 & 0 & 0 & 0 & 0 & 0 & 0 & 0\\
    13 & 0 & 0 & 0 & 0 & 0 & 0 & 0 & \textcolor{blue}{1} & \textcolor{blue}{1} & 0 & 0 & 0 & 0 & 0 & 0 & 0 & 0 & 0\\
    14 & 0 & 0 & 0 & 0 & 0 & 0 & 0 & 0 & \textcolor{blue}{1} & 0 & 0 & 0 & 0 & 0 & 0 & \textcolor{blue}{1} & 0 & 0\\
    15 & 0 & 0 & 0 & 0 & 0 & 0 & \textcolor{blue}{1} & 0 & 0 & \textcolor{blue}{1} & 0 & 0 & 0 & 0 & 0 & 0 & 0 & 0\\
    16 & 0 & 0 & 0 & 0 & 0 & 0 & 0 & \textcolor{blue}{1} & 0 & 0 & \textcolor{blue}{1} & 0 & 0 & 0 & 0 & 0 & 0 & 0\\
    17 & 0 & 0 & 0 & 0 & 0 & 0 & 0 & 0 & \textcolor{blue}{1} & 0 & 0 & \textcolor{blue}{1} & 0 & 0 & 0 & 0 & 0 & 0\\
    18 & 0 & 0 & 0 & 0 & 0 & 0 & 0 & 0 & 0 & \textcolor{blue}{1} & \textcolor{blue}{1} & 0 & 0 & 0 & 0 & 0 & 0 & 0\\
    19 & 0 & 0 & 0 & 0 & 0 & 0 & 0 & 0 & 0 & 0 & \textcolor{blue}{1} & \textcolor{blue}{1} & 0 & 0 & 0 & 0 & 0 & 0\\
    20 & \textcolor{blue}{1} & 0 & 0 & 0 & 0 & 0 & 0 & 0 & 0 & 0 & 0 & \textcolor{blue}{1} & 0 & 0 & 0 & 0 & 0 & 0\\
    21 & 0 & 0 & 0 & 0 & 0 & 0 & 0 & 0 & 0 & \textcolor{blue}{1} & 0 & 0 & \textcolor{blue}{1} & 0 & 0 & 0 & 0 & 0\\
    22 & 0 & 0 & 0 & 0 & 0 & 0 & 0 & 0 & 0 & 0 & \textcolor{blue}{1} & 0 & 0 & \textcolor{blue}{1} & 0 & 0 & 0 & 0\\
    23 & 0 & 0 & 0 & 0 & 0 & 0 & 0 & 0 & 0 & 0 & 0 & \textcolor{blue}{1} & 0 & 0 & \textcolor{blue}{1} & 0 & 0 & 0\\
    24 & 0 & 0 & 0 & 0 & 0 & 0 & 0 & 0 & 0 & 0 & 0 & 0 & \textcolor{blue}{1} & \textcolor{blue}{1} & 0 & 0 & 0 & 0\\
    25 & 0 & 0 & 0 & 0 & 0 & 0 & 0 & 0 & 0 & 0 & 0 & 0 & 0 & \textcolor{blue}{1} & \textcolor{blue}{1} & 0 & 0 & 0\\
    26 & 0 & 0 & 0 & \textcolor{blue}{1} & 0 & 0 & 0 & 0 & 0 & 0 & 0 & 0 & 0 & 0 & \textcolor{blue}{1} & 0 & 0 & 0\\
    27 & 0 & 0 & 0 & 0 & 0 & 0 & 0 & 0 & 0 & 0 & 0 & 0 & \textcolor{blue}{1} & 0 & 0 & \textcolor{blue}{1} & 0 & 0\\
    28 & 0 & 0 & 0 & 0 & 0 & 0 & 0 & 0 & 0 & 0 & 0 & 0 & 0 & \textcolor{blue}{1} & 0 & 0 & \textcolor{blue}{1} & 0\\
    29 & 0 & 0 & 0 & 0 & 0 & 0 & 0 & 0 & 0 & 0 & 0 & 0 & 0 & 0 & \textcolor{blue}{1} & 0 & 0 & \textcolor{blue}{1}\\
    30 & 0 & 0 & 0 & 0 & 0 & 0 & 0 & 0 & 0 & 0 & 0 & 0 & 0 & 0 & 0 & \textcolor{blue}{1} & \textcolor{blue}{1} & 0\\
    31 & 0 & 0 & 0 & 0 & 0 & 0 & 0 & 0 & 0 & 0 & 0 & 0 & 0 & 0 & 0 & 0 & \textcolor{blue}{1} & \textcolor{blue}{1}\\
    32 & 0 & 0 & 0 & 0 & 0 & 0 & \textcolor{blue}{1} & 0 & 0 & 0 & 0 & 0 & 0 & 0 & 0 & 0 & 0 & \textcolor{blue}{1}\\
    33 & \textcolor{blue}{1} & 0 & 0 & 0 & 0 & 0 & 0 & 0 & 0 & 0 & 0 & 0 & 0 & 0 & 0 & \textcolor{blue}{1} & 0 & 0\\
    34 & 0 & \textcolor{blue}{1} & 0 & 0 & 0 & 0 & 0 & 0 & 0 & 0 & 0 & 0 & 0 & 0 & 0 & 0 & \textcolor{blue}{1} & 0\\
    35 & 0 & 0 & \textcolor{blue}{1} & 0 & 0 & 0 & 0 & 0 & 0 & 0 & 0 & 0 & 0 & 0 & 0 & 0 & 0 & \textcolor{blue}{1}\\
    \end{array}
    \nonumber
% \label{eq: Av matrix}
\end{eqs}
where the rows and columns of the matrix correspond to qubit indices and shifted $A_v$ stabilizers, respectively. For example, the $A_v$ stabilizer at the origin in Fig.~\ref{fig: number_convention a} applies Pauli $X$ to qubits $9$, $12$, $15$, and $32$, which corresponds to the $7$th column of the parity-check matrix above. On a lattice with 36 qubits, there are 18 distinct $A_v$ stabilizers; however, their product yields the identity operator, rendering one stabilizer redundant and giving $\operatorname{rank}(H_X) = 17$.
The shifts defined in Eqs.~\eqref{eq: shift ly in y} and \eqref{eq: shift lx in x} generalize the ``cyclic structure'' of parity-check matrices in hypergraph product codes~\cite{tillich2013quantum, kovalev2012improved}.

The $36 \times 18$ parity-check matrix $H_Z$ associated with the $B_p$ stabilizers is constructed analogously, following the same procedure outlined above:
\begin{eqs}
    \renewcommand{\arraystretch}{1.0}
    \begin{array}{C|*{18}{C}}        
       & 1 & 2 & 3 & 4 & 5 & 6 & 7 & 8 & 9 & 10 & 11 & 12 & 13 & 14 & 15 & 16 & 17 & 18\\\hline
     0 & \textcolor{blue}{1} & 0 & 0 & 0 & 0 & 0 & 0 & 0 & 0 & 0 & 0 & 0 & 0 & 0 & 0 & \textcolor{blue}{1} & 0 & 0\\
     1 & 0 & \textcolor{blue}{1} & 0 & 0 & 0 & 0 & 0 & 0 & 0 & 0 & 0 & 0 & 0 & 0 & 0 & 0 & \textcolor{blue}{1} & 0\\
     2 & 0 & 0 & \textcolor{blue}{1} & 0 & 0 & 0 & 0 & 0 & 0 & 0 & 0 & 0 & 0 & 0 & 0 & 0 & 0 & \textcolor{blue}{1}\\
     3 & \textcolor{blue}{1} & 0 & 0 & 0 & 0 & 0 & 0 & 0 & 0 & 0 & 0 & \textcolor{blue}{1} & 0 & 0 & 0 & 0 & 0 & 0\\
     4 & \textcolor{blue}{1} & \textcolor{blue}{1} & 0 & 0 & 0 & 0 & 0 & 0 & 0 & 0 & 0 & 0 & 0 & 0 & 0 & 0 & 0 & 0\\
     5 & 0 & \textcolor{blue}{1} & \textcolor{blue}{1} & 0 & 0 & 0 & 0 & 0 & 0 & 0 & 0 & 0 & 0 & 0 & 0 & 0 & 0 & 0\\
     6 & \textcolor{blue}{1} & 0 & 0 & \textcolor{blue}{1} & 0 & 0 & 0 & 0 & 0 & 0 & 0 & 0 & 0 & 0 & 0 & 0 & 0 & 0\\
     7 & 0 & \textcolor{blue}{1} & 0 & 0 & \textcolor{blue}{1} & 0 & 0 & 0 & 0 & 0 & 0 & 0 & 0 & 0 & 0 & 0 & 0 & 0\\
     8 & 0 & 0 & \textcolor{blue}{1} & 0 & 0 & \textcolor{blue}{1} & 0 & 0 & 0 & 0 & 0 & 0 & 0 & 0 & 0 & 0 & 0 & 0\\
    9 & 0 & 0 & 0 & \textcolor{blue}{1} & 0 & 0 & 0 & 0 & 0 & 0 & 0 & 0 & 0 & 0 & \textcolor{blue}{1} & 0 & 0 & 0\\
    10 & 0 & 0 & 0 & \textcolor{blue}{1} & \textcolor{blue}{1} & 0 & 0 & 0 & 0 & 0 & 0 & 0 & 0 & 0 & 0 & 0 & 0 & 0\\
    11 & 0 & 0 & 0 & 0 & \textcolor{blue}{1} & \textcolor{blue}{1} & 0 & 0 & 0 & 0 & 0 & 0 & 0 & 0 & 0 & 0 & 0 & 0\\
    12 & 0 & 0 & 0 & \textcolor{blue}{1} & 0 & 0 & \textcolor{blue}{1} & 0 & 0 & 0 & 0 & 0 & 0 & 0 & 0 & 0 & 0 & 0\\
    13 & 0 & 0 & 0 & 0 & \textcolor{blue}{1} & 0 & 0 & \textcolor{blue}{1} & 0 & 0 & 0 & 0 & 0 & 0 & 0 & 0 & 0 & 0\\
    14 & 0 & 0 & 0 & 0 & 0 & \textcolor{blue}{1} & 0 & 0 & \textcolor{blue}{1} & 0 & 0 & 0 & 0 & 0 & 0 & 0 & 0 & 0\\
    15 & 0 & 0 & 0 & 0 & 0 & 0 & \textcolor{blue}{1} & 0 & 0 & 0 & 0 & 0 & 0 & 0 & 0 & 0 & 0 & \textcolor{blue}{1}\\
    16 & 0 & 0 & 0 & 0 & 0 & 0 & \textcolor{blue}{1} & \textcolor{blue}{1} & 0 & 0 & 0 & 0 & 0 & 0 & 0 & 0 & 0 & 0\\
    17 & 0 & 0 & 0 & 0 & 0 & 0 & 0 & \textcolor{blue}{1} & \textcolor{blue}{1} & 0 & 0 & 0 & 0 & 0 & 0 & 0 & 0 & 0\\
    18 & 0 & 0 & 0 & 0 & 0 & 0 & \textcolor{blue}{1} & 0 & 0 & \textcolor{blue}{1} & 0 & 0 & 0 & 0 & 0 & 0 & 0 & 0\\
    19 & 0 & 0 & 0 & 0 & 0 & 0 & 0 & \textcolor{blue}{1} & 0 & 0 & \textcolor{blue}{1} & 0 & 0 & 0 & 0 & 0 & 0 & 0\\
    20 & 0 & 0 & 0 & 0 & 0 & 0 & 0 & 0 & \textcolor{blue}{1} & 0 & 0 & \textcolor{blue}{1} & 0 & 0 & 0 & 0 & 0 & 0\\
    21 & 0 & 0 & \textcolor{blue}{1} & 0 & 0 & 0 & 0 & 0 & 0 & \textcolor{blue}{1} & 0 & 0 & 0 & 0 & 0 & 0 & 0 & 0\\
    22 & 0 & 0 & 0 & 0 & 0 & 0 & 0 & 0 & 0 & \textcolor{blue}{1} & \textcolor{blue}{1} & 0 & 0 & 0 & 0 & 0 & 0 & 0\\
    23 & 0 & 0 & 0 & 0 & 0 & 0 & 0 & 0 & 0 & 0 & \textcolor{blue}{1} & \textcolor{blue}{1} & 0 & 0 & 0 & 0 & 0 & 0\\
    24 & 0 & 0 & 0 & 0 & 0 & 0 & 0 & 0 & 0 & \textcolor{blue}{1} & 0 & 0 & \textcolor{blue}{1} & 0 & 0 & 0 & 0 & 0\\
    25 & 0 & 0 & 0 & 0 & 0 & 0 & 0 & 0 & 0 & 0 & \textcolor{blue}{1} & 0 & 0 & \textcolor{blue}{1} & 0 & 0 & 0 & 0\\
    26 & 0 & 0 & 0 & 0 & 0 & 0 & 0 & 0 & 0 & 0 & 0 & \textcolor{blue}{1} & 0 & 0 & \textcolor{blue}{1} & 0 & 0 & 0\\
    27 & 0 & 0 & 0 & 0 & 0 & \textcolor{blue}{1} & 0 & 0 & 0 & 0 & 0 & 0 & \textcolor{blue}{1} & 0 & 0 & 0 & 0 & 0\\
    28 & 0 & 0 & 0 & 0 & 0 & 0 & 0 & 0 & 0 & 0 & 0 & 0 & \textcolor{blue}{1} & \textcolor{blue}{1} & 0 & 0 & 0 & 0\\
    29 & 0 & 0 & 0 & 0 & 0 & 0 & 0 & 0 & 0 & 0 & 0 & 0 & 0 & \textcolor{blue}{1} & \textcolor{blue}{1} & 0 & 0 & 0\\
    30 & 0 & 0 & 0 & 0 & 0 & 0 & 0 & 0 & 0 & 0 & 0 & 0 & \textcolor{blue}{1} & 0 & 0 & \textcolor{blue}{1} & 0 & 0\\
    31 & 0 & 0 & 0 & 0 & 0 & 0 & 0 & 0 & 0 & 0 & 0 & 0 & 0 & \textcolor{blue}{1} & 0 & 0 & \textcolor{blue}{1} & 0\\
    32 & 0 & 0 & 0 & 0 & 0 & 0 & 0 & 0 & 0 & 0 & 0 & 0 & 0 & 0 & \textcolor{blue}{1} & 0 & 0 & \textcolor{blue}{1}\\
    33 & 0 & 0 & 0 & 0 & 0 & 0 & 0 & 0 & \textcolor{blue}{1} & 0 & 0 & 0 & 0 & 0 & 0 & \textcolor{blue}{1} & 0 & 0\\
    34 & 0 & 0 & 0 & 0 & 0 & 0 & 0 & 0 & 0 & 0 & 0 & 0 & 0 & 0 & 0 & \textcolor{blue}{1} & \textcolor{blue}{1} & 0\\
    35 & 0 & 0 & 0 & 0 & 0 & 0 & 0 & 0 & 0 & 0 & 0 & 0 & 0 & 0 & 0 & 0 & \textcolor{blue}{1} & \textcolor{blue}{1}\\
    \end{array}
    \nonumber
% \label{eq: Bp matrix}
\end{eqs}
%\end{widetext}
where the entries of $H_Z$ are defined such that a value of $1$ indicates the qubit positions on which Pauli $Z$ acts. For example, the $B_p$ stabilizer shown in Fig.~\ref{fig: number_convention a} applies Pauli $Z$ to qubits $3$, $20$, $23$, and $26$, corresponding to the $12$th column of $H_Z$. Since the product of all $B_p$ stabilizers is the identity, one of them is redundant, leading to $\operatorname{rank}(H_Z) = 17$.

Therefore, the logical dimension of the Kitaev toric code is
\begin{eqs}
    k &= n - \operatorname{rank}(H_X) - \operatorname{rank}(H_Z) \\
      &= 36 - 17 - 17 = 2,
\end{eqs}
in agreement with Example~\ref{example: Kitaev toric code}.

%\end{widetext}

%\begin{widetext}
    
\newpage
\section{Algorithm}

In this appendix, we present the pseudocode for computing the logical dimension $k$, as stated in Theorem~\ref{thm: k on twisted torus} of the main text.

\begin{algorithm}[H]
    \caption{Deriving the logical dimension $k$}
    \begin{algorithmic}[1]
        \REQUIRE Stabilizer polynomials $\mathcal{S}$.
        \ENSURE The logical dimension $k$ computed via the Gr\"obner basis.

        \STATE Consider the stabilizer code defined by:
        \begin{equation}
            A_v = 
            \begin{bmatrix}
                f(x,y) \\
                \rule{0pt}{1.1em}g(x,y) \\
                \hline
                0 \\
                0
            \end{bmatrix}, 
            \quad
            B_p = 
            \begin{bmatrix}
                0 \\
                0 \\
                \hline
                \rule{0pt}{1.1em}\overline{g(x,y)} \\
                \overline{f(x,y)}
            \end{bmatrix},
        \end{equation}
        defined on a twisted torus with lattice vectors $\vec{a}_1 = (0, \alpha)$ and $\vec{a}_2 = (\beta, \gamma)$.

        \STATE Compute the Gr\"obner basis $\langle h(x,y), i(x,y), j(x,y) \rangle$:
        \begin{equation}
            \langle f(x,y), g(x,y), y^{\alpha} - 1, x^\beta y^{\gamma} - 1 \rangle = \langle h(x,y), i(x,y), j(x,y) \rangle,
        \end{equation}
        using the \texttt{groebner} function from the Python package \texttt{sympy}.

        \IF{$\langle h(x,y), i(x,y), j(x,y) \rangle = \langle 1 \rangle$}
            \RETURN $k=0$
        \ELSE
        \STATE Calculate the number of independent monomials:
        \STATE Determine the maximal degrees $n_x$ and $n_y$ of the variables $x$ and $y$ in the polynomials $h(x,y)$, $i(x,y)$, and $j(x,y)$.
        \STATE Construct matrix $\wwide{M}_1$ via the translation duplication map $\mathrm{TD}_{m_x,m_y}$ defined in Ref.~\cite{liang2023extracting}:
        \begin{equation}
            \wwide{M}_1 \leftarrow 
            \begin{bmatrix}
                \wwide{\mathrm{TD}_{\frac{n_x}{2},\frac{n_y}{2}}(1)} \\
                \wwide{\mathrm{TD}_{\frac{3n_x}{2},\frac{3n_y}{2}}(h(x,y))} \\
                \wwide{\mathrm{TD}_{\frac{3n_x}{2},\frac{3n_y}{2}}(i(x,y))} \\
                \wwide{\mathrm{TD}_{\frac{3n_x}{2},\frac{3n_y}{2}}(j(x,y))}
            \end{bmatrix}.
        \end{equation}

        \STATE Construct matrix $\wwide{M}_2$ via the translation duplication map:
        \begin{equation}
            \wwide{M}_2 \leftarrow 
            \begin{bmatrix}
                \wwide{\mathrm{TD}_{\frac{3n_x}{2},\frac{3n_y}{2}}(h(x,y))} \\
                \wwide{\mathrm{TD}_{\frac{3n_x}{2},\frac{3n_y}{2}}(i(x,y))} \\
                \wwide{\mathrm{TD}_{\frac{3n_x}{2},\frac{3n_y}{2}}(j(x,y))}
            \end{bmatrix}.
        \end{equation}

        \RETURN
        \begin{equation}
            k = 2 \times \left(\mathrm{rank}(\wwide{M}_1) - \mathrm{rank}(\wwide{M}_2)\right)
        \end{equation}

        \ENDIF

    \end{algorithmic}
\end{algorithm}

\end{widetext}

%%%%%%%%%%%%%%%%%%%%%%%%%%%%%%%%%%%%%%%%%%%%%%%%%%%%%%%%%%%%%%%%%%%%%%%%%%%%%%%%%%%%
\bibliography{bib.bib}

\end{document}